\documentclass[oneside,english]{amsart}
\usepackage[T1]{fontenc}
\usepackage[latin9]{inputenc}
\usepackage{geometry}
\usepackage{color}
\usepackage{float}
\usepackage{mathrsfs}
\usepackage{mathtools}
\usepackage{amstext}
\usepackage{amsthm}
\usepackage{amssymb}
\usepackage{graphicx}
\newtheorem{corollary}{Corollary}
\newtheorem{definition}{Definition}

\makeatletter

\providecommand{\tabularnewline}{\\}

\numberwithin{equation}{section}
\theoremstyle{plain}
\newtheorem{thm}{\protect\theoremname}[section]
\theoremstyle{definition}
\newtheorem{defn}[thm]{\protect\definitionname}
\theoremstyle{plain}
\newtheorem{prop}[thm]{\protect\propositionname}
\theoremstyle{remark}
\newtheorem{rem}[thm]{\protect\remarkname}
\theoremstyle{remark}
\newtheorem*{rem*}{\protect\remarkname}
\theoremstyle{plain}

\theoremstyle{remark}
\newtheorem{claim}[thm]{\protect\claimname}
\theoremstyle{plain}
\newtheorem{cor}[thm]{\protect\corollaryname}
\theoremstyle{remark}
\newtheorem*{acknowledgement*}{\protect\acknowledgementname}

\makeatother

\usepackage{babel}
\usepackage{amsmath,amssymb}
\providecommand{\acknowledgementname}{Acknowledgement}
\providecommand{\claimname}{Claim}
\providecommand{\conjecturename}{Conjecture}
\providecommand{\corollaryname}{Corollary}
\providecommand{\definitionname}{Definition}
\providecommand{\propositionname}{Proposition}
\providecommand{\remarkname}{Remark}
\providecommand{\theoremname}{Theorem}

\begin{document}
\title[Causal Classification of Misner-Type Spacetimes]{Causal Classification of Pathological Misner-Type Spacetimes}

\author{
N. E. Rieger \\ 
\textnormal{Mathematics Department, University of California Irvine, Rowland Hall, 92697 Irvine, USA} \\ 
\textnormal{Current address: Department of Mathematics, Yale University, 219 Prospect Street, New Haven, CT 06520, USA} \\ 
\textnormal{\texttt{n.rieger@yale.edu}}
}

~

~

\begin{abstract}
We investigate three causality-violating spacetimes: Misner space (including Kip Thorne's ``moving wall'' model), 
the pseudo-Schwarzschild spacetime, and a new model introduced here, the pseudo-Reissner-Nordstr\"{o}m spacetime. 
Despite their different physical origins---ranging from a flat vacuum solution to a black-hole-type vacuum solution 
to a non-vacuum model requiring exotic matter---all three share a common warped-product structure, 
$2$-dimensional cylindrical base metrics of Eddington-Finkelstein type, and fundamental causal features such as 
Cauchy and chronology horizons, acausal regions, and analogous geodesic behaviour. 
Building on a conjecture first proposed in 2016, we present a formal proof that the three models are pairwise isocausal
on their universal covers and on suitable causally regular regions of their compactified forms. 
The proof is constructive, providing explicit causal bijections on the covers and identifying a concrete 
deck-equivariance criterion governing descent to the compactified spacetimes: if the equivariance degree satisfies $|k|=1$ 
the models are globally isocausal, whereas if $|k|>1$ or equivariance fails, then at most a one-way causal relation 
holds between the compactified models. 
These results supply a rigorous causal classification linking these spacetimes, 
placing them within a unified Misner-type family and providing a framework for extending the classification 
to other causality-violating solutions.
\end{abstract}

\maketitle
~\\ \textbf{This article is based on research originally conducted as part\\ of a project during 2016--2018 under the supervision of Kip S. Thorne.}\\

~\\ 

\section{Introduction}

General relativity defines a class of cosmological models, each representing
an idealization of a physically possible universe compatible with the theory.
A cosmological model is a mathematical description of this idealized universe, 
referred to as a spacetime and represented by a pair $(M,g)$, where $M$ is a differentiable manifold 
capturing the topology and continuity of the universe, and $g$ is a smooth, symmetric, non-degenerate 
$(0,2)$-tensor field encoding the geometric and causal structure. 
Each point in $M$ represents an event, and the causal structure determined by $g$ 
governs the possible trajectories of particles and light rays. 
Within this framework, the theoretical possibility of time travel---through 
closed timelike curves (CTCs) that allow a timelike observer to return to an event in their own past---has been 
studied extensively~\cite{Awad, Emparan, Gavassino, Kim and Thorne, Roy, Thorne CTC, Thorne}. 
In 1967, Charles Misner introduced the Misner space~\cite{Misner} as a minimal example of a spacetime 
with a chronology-violating region. While constructing spacetimes containing CTCs is relatively straightforward, 
finding examples with a plausible physical interpretation is more challenging. 
The pseudo-Schwarzschild spacetime, proposed by Ori~\cite{Ori} in 2007, is a more sophisticated 
causality-violating model that incorporates features such as Cauchy horizons, avoids singularities, 
and models the breakdown of determinism inside black holes.\footnote{An ``idealized black hole model'' is a 
simplified theoretical depiction of a black hole, commonly assuming perfect spherical symmetry and no rotation, 
and typically derived from the Schwarzschild metric.}

\medskip
We begin by examining these two apparently different spacetimes---Misner space (also in the form of Kip Thorne's ``moving wall'' model) and the pseudo-Schwarzschild spacetime---to clarify their geometric structures and causal properties,
with a particular focus on pathologies. The two are linked by their chronological and global features. 
Indeed, the pseudo-Schwarzschild spacetime can be viewed as a non-flat generalization of Misner space: 
in the limit as the mass parameter $M \to 0$, the pseudo--Schwarzschild metric reduces to the Misner metric. 
In earlier work~\cite{Rieger - Topologies of maximally extended non-Hausdorff Misner Space} we conjectured that these 
spacetimes are isocausal in the sense of Garc\'{i}a-Parrado and Senovilla.

\medskip
We then introduce a new chronology-violating spacetime: the pseudo-Reissner-Nordstr\"{o}m spacetime, 
which generalizes the previous two by including an electric charge parameter in analogy with the classical Reissner-Nordstr\"{o}m solution. Like its predecessors, this model possesses Cauchy and chronology horizons, 
acausal regions, and a warped-product structure with a $2$-dimensional cylindrical base and hyperbolic spatial fibers. 
Its inclusion is especially noteworthy because it is a non-vacuum solution, in contrast to the vacuum nature of Misner 
and pseudo-Schwarzschild spacetimes, yet it still exhibits the same characteristic causal pathologies.

\medskip
While the spacetimes considered in this work are mathematically well-defined solutions of Einstein's equations, some of them---most notably the pseudo-Reissner-Nordstr\"{o}m model---require violations of classical energy conditions. Consequently, their direct physical realizability in a classical astrophysical context is a subject of ongoing debate. Nevertheless, the investigation of such causality-violating geometries is well motivated for several reasons. First, they serve as controlled testbeds for probing the limits of general relativity, particularly in regimes where classical notions of causality and determinism break down. Second, these models provide insight into the structure and stability of Cauchy horizons, which are central features of physically relevant solutions such as the interior of charged or rotating black holes. Third, the study of spacetimes admitting closed timelike curves is closely related to foundational questions, including the validity of the chronology protection conjecture proposed by Stephen Hawking, and the role of quantum effects in enforcing causal consistency. In particular, such examples help clarify the role of classical energy conditions by exhibiting the types of pathological behavior---such as closed timelike curves or breakdowns of predictability---that these conditions are intended to exclude. From this perspective, even spacetimes supported by exotic matter can illuminate which features of causality violation are generic and which depend sensitively on the underlying matter content. The unified framework developed in this work therefore contributes not only to the classification of specific models, but also to a broader understanding of the interplay between geometry, causality, and physical viability in general relativity.

\medskip
A central contribution of this paper is to replace the earlier conjecture with a formal proof that all three 
spacetimes---Misner, pseudo-Schwarzschild, and pseudo-Reissner-Nordstr\"{o}m---are pairwise isocausal 
on their universal covers, and on suitable open, causally regular regions of the compactified spacetimes. 
Proposition~\ref{prop:IsoAllThree} provides explicit causal bijections between the models at the covering level, 
and identifies a concrete deck-equivariance criterion governing descent to the compactified spacetimes: 
when the equivariance degree satisfies $|k|=1$ the models are globally isocausal, 
whereas if $|k|>1$ or if equivariance fails, then only a one-way causal relation is obtained on the quotient. 
Corollary~\ref{cor:NoEquivariance} and Remark~\ref{rem:DefaultOneWay} emphasize that this one-way case is the 
generic global situation. This sharpens the classification of these causality-violating models and establishes, 
for the first time, a rigorous causal equivalence framework connecting them.

\medskip
From a mathematical perspective, the shared warped-product structure and conformal relation of the base metrics 
make these isocausality results natural; from a physical perspective, they suggest that key causal features may arise independently of the detailed matter content. This work provides a unified geometric and causal framework for understanding these examples, and 
paves the way for extending the classification to other models---such as the pseudo-Kerr spacetime---that may 
also belong to this Misner-type, causality-violating family.

\subsection{Causality violating spacetimes}

Before presenting the main results we want to review some useful definitions. 

First, recall that a Lorentzian manifold $M$ is said to be \emph{past-distinguishing}
if any two points $p,q\in M$ with identical chronological pasts must coincide; that is, $I^{-}(p)=I^{-}(q)\,\Longrightarrow p=q$. Similarly, \( M \) is said to be \emph{future-distinguishing} if any
two points $p,\,q\in M$ the same chronological future must be identical: $I^{+}(p)=I^{+}(q)\Longrightarrow p=q$. 
We say a spacetime satisfies the \emph{distinguishability condition} if it is both, past-
and future-distinguishing. Thus, a Lorentzian manifold is
\emph{not} past-distinguishing if distinct points can have identical
pasts.  A spacetime allows time travel if a point \( p \in M \) lies in its own timelike future, i.e., \( p \in I^{+}(p) \), resulting in a closed timelike curve that violates causality. Hence, for distinct points \( p, q \in M \), it is possible to have \( p \ll q \ll p \). The \emph{chronology violating region} of a spacetime \( (M, g) \) is defined as \( \mathcal{V}(M) = \{ p \in M : p \in I^{+}(p) \} \).\footnote{The set \( \mathcal{V}(M) \) decomposes into equivalence classes \( [p] \), where \( p \sim q \) if \( p \ll q \ll p \). That is, two points \( p, q \in [p] \) belong to the same equivalence class if there exists a closed timelike curve \( \gamma \) connecting \( p \) and \( q \) with \( \gamma(I) \subset [p] \). Points within the same equivalence class share the same chronological character, making it impossible to establish a definite chronological order of events among them.}

\begin{defn}
A spacetime \( (M,g) \) is said to be \emph{chronological} and to satisfy the \emph{chronology condition} if it does not contain any closed timelike curves; that is, \( p \notin I^{+}(p) \) for all \( p \in M \).
\end{defn}

\begin{prop}
\cite{Goehring} If the spacetime $(M,g)$ is compact then the chronology
violating region is not empty and contains a closed timelike curve.
\end{prop}


 
\begin{defn}
A subset \( S \subset M \) is said to be \emph{achronal} if no timelike curve intersects it more than once; that is, there do not exist points \( p, q \in S \) such that \( q \in I^{+}(p) \). Equivalently, \( I^{+}(S) \cap S = \varnothing \).
\end{defn}


The boundaries separating chronal and achronal regions are called
\emph{chronology horizons}. The chronal region ends and closed timelike
curves arise at the future chronology horizon. Conversely, closed
timelike curves vanish and the chronal region begins at the past chronology
horizon. A future chronology horizon is a special type of future Cauchy horizon
and it is therefore subject to all the properties of such horizons
\cite{Frolov}. The \emph{past chronology horizon} is defined dually.


\begin{defn}
Let $J^{-}(U)$ denote the causal past of a set $U\subset M$, and
let $\bar{J}^{-}(U)$ represent the topological closure of $J^{-}$.
Define $\mathcal{I}^{+}$ and $\mathcal{I}^{-}$ as future and past
null infinity, respectively. The \emph{boundary} of $\bar{J}^{-}$
is given by $\partial\bar{J}^{-}(U)=\bar{J}^{-}(U)\setminus J^{-}(U)$,
and the \emph{future event horizon} of $M$ is $\mathcal{H}^{+}=\partial\bar{J}^{-}(\mathcal{I}^{+})$.
\end{defn}


The event horizon is often referred to as the ``point of no return''.
It is useful to define $\bar{J}^{-}(\mathcal{I}^{+})$ as the \emph{domain
of outer communications}. The complement $M\setminus\bar{J}^{-}(\mathcal{I}^{+})$
is then referred to as the \emph{black hole} (region). A spacetime
may contain a singularity in several senses. In this context, a singularity
is a hypersurface on which all worldlines that pass through the event
horizon terminate.

\section{Misner space\label{sec:Misner space}}

Misner space~\cite{Misner}, typically presented as a $2$-dimensional
toy model, is a type of spacetime in general relativity that that
illustrates how singularities might form. This spacetime begins with
causally well-behaved initial conditions but later develops closed
timelike curves, leading to a chronology-violating region. Characterized
by a conformally flat metric, Misner space is also Ricci-flat, making
it a valid vacuum solution to Einstein's field equations in the absence
of matter and energy. The $4$-dimensional version can be obtained
as a straightforward extension.  It can be considered as Minkowski
space with an altered topology due to identification under a boost.\footnote{A Lorentz boost is a Lorentz transformation with no rotation.}\textcolor{black}{{}
To see this, start with the Lorentz-Minkowski spacetime $\mathcal{\mathscr{M}}:=(\mathbb{R}^{4},\tilde{\eta})$,
defined as the smooth manifold $\mathbb{R}^{1,3}:=(\mathbb{R}^{4},\tilde{\eta})$
equipped with Lorentzian coordinates $(\tilde{t},\tilde{x}_{1},\tilde{x}_{2},\tilde{x}_{3})$.
These coordinates yield the metric tensor }\textcolor{black}{\emph{$\tilde{\eta}$}}\textcolor{black}{,
given by}

\begin{equation}
ds^{2}:=-d\tilde{t}{}^{2}+(d\tilde{x}_{1})^{2}+(d\tilde{x}_{2})^{2}+(d\tilde{x}_{3})^{2}.\label{eq:minkowski metric}
\end{equation}
~\\ This coordinate system for Minkowski spacetime is related to
Misner coordinates $(\eta,X_{1},X_{2},X_{3})$ through the following
coordinate transformation

\begin{equation}
\tilde{t}=2\sqrt{\eta}\cosh(\frac{X_{1}}{2}),\;\tilde{x}_{1}=2\sqrt{\eta}\sinh(\frac{X_{1}}{2}),\;\tilde{x}_{2}=X_{2},\;\tilde{x}_{3}=X_{3}.
\end{equation}

~\\ This coordinate change results in the Misner metric $ds^{2}=-\eta^{-1}d\eta^{2}+\eta(dX_{1})^{2}+(dX_{2})^{2}+(dX_{3})^{2}$,
where $0<\eta<\infty$, $0\leq X_{1}\leq2\pi$ as an angular coordinate, and $-\infty\leq X_{2}$,
$X_{3}\leq\infty$. Due to the periodic nature of the coordinate $X_{1}$,
the underlying topology is given by $S^{1}\times\mathbb{R}^{3}$.
However, at $\eta=0$ the metric exhibits an apparent (coordinate)
singularity. To address this, we introduce a new set of coordinates
to extend the metric beyond $\eta=0$. Specifically, we define a new coordinate $\varphi=X_{1}-\ln(\eta)$.
To transform the metric accordingly, we also require $d\varphi=dX_{1}+\frac{d\eta}{\eta}$.
Thus, the metric transforms into $ds^{2}=2d\eta d\varphi+\eta d\varphi^{2}+(dX_{2})^{2}+(dX_{3})^{2}$. Finally, since $\eta$ is a timelike coordinate, we introduce the substitution $T=\eta$, leading to the metric

\begin{equation}
ds^{2}=2dTd\varphi+Td\varphi{}^{2}+(dX_{2})^{2}+(dX_{3})^{2},\label{eq:Misner extension1}
\end{equation}

~\\ where $T$ is a timelike coordinate and $\varphi$ is an angular
coordinate, with domains $-\infty<T<\infty$ and $0\leq\varphi\leq a$.\footnote{Due to the periodicity condition, $\varphi$ can take any real number
value. Specifically, $\varphi$ and $\varphi+na$ are equivalent for
any integer $n$.}
Note that this is only one of two possible inequivalent extensions
of Misner space. Alternatively, we could start with the Minkowski
metric~(\ref{eq:minkowski metric}) again, but apply a different
coordinate transformation, subsequently extending the resulting Misner
space across the coordinate singularity in a different way. In both
cases, we obtain a cylindrical Misner spacetime defined on $M=S^{1}\times\mathbb{R}^{3}$. 

~

Accordingly, Minkowski space $\mathscr{M}$ serves as the associated
covering space of Misner space, where identical points are determined
by the identification 

\begin{equation}
(\tilde{t},\tilde{x}_{1},\tilde{x}_{2},\tilde{x}_{3})\leftrightarrow(\tilde{t}\cosh(na)+\tilde{x}_{1}\sinh(na),\tilde{t}\sinh(na)+\tilde{x}_{1}\cosh(na),\tilde{x}_{2},\tilde{x}_{3}),
\end{equation}
where $a$ is the period associated with the periodic direction $X_{1}$.\footnote{Note, that these points correspond to $(\eta,X_{1},X_{2},X_{3})\leftrightarrow(\eta,X_{1}+na,X_{2},X_{3})$
in Misner coordinates.}

~

We orient Misner space in time by requiring $-\partial_{T}$ to point toward the future.
The vector field $\partial_{T}$ iis everywhere lightlike, 
while $\partial_{\varphi}$ is spacelike for $T>0$, lightlike at
$T=0$, and timelike for $T<0$. At \( T = 0 \), there exists a null surface known 
as the  \emph{chronology horizon}, which separates the chronology-violating 
region from the chronal region. This surface is intersected exactly 
once by every causal curve, thereby serving as the Cauchy horizon for 
any initial hypersurface with $T_{0}=\textrm{const}>0$.
Hypersurfaces of constant $T=\textrm{const}>0$ are spacelike and and can 
be chosen as initial data hypersurfaces for specifying the evolution of fields.

\begin{prop}
If a closed timelike curve (CTC) arises in a chronology-violating
spacetime, then it contains an unlimited number of CTCs.\label{pro: infinite CTCs}
\end{prop}

\begin{proof}
Let $\alpha$ be a closed timelike curve in the spacetime $M$. Then
consider two points $p$ and $q$ on $\alpha$ that satisfy the condition
$p\ll q$. Consequently, we have $p\ll q\ll p$ and \emph{$p\in I^{+}(p)$
}as well as\emph{ $q\in I^{+}(q)$. }The same applies for the timelike
past. The sets\emph{ $I^{+}(p)$ }and\emph{ $I^{-}(q)$ }are open
and their finite intersection\emph{ $O:=I^{+}(p)\cap I^{-}(q)$ }is
therefore not empty and an open set as well. Choose now any $r\in O$.
Since $r\in I^{+}(p)$ there is a future directed timelike curve segment
from $p$ to $r$, and because $r\in I^{-}(q)$, there is a future
directed timelike curve segment from $q$ to $r$. It remains to join
the two curve segments at the point $r$ which can be done smoothly
due to the fact that $O$ is an open neighborhood of $r$.
\end{proof}
The location at which closed timelike curves first emerge is called
the \emph{chronology horizon. }It separates a chronal region (without
CTCs) from a non-chronal region that contains CTCs at every point.
It is a well known fact that Misner space contains closed timelike
curves, but above all it follows from Section~\ref{subsec:Misner-Geodesics+warped product}
below, Misner space also comprises closed timelike geodesics.

\subsection{Kip Thorne's physics approach to Misner space\label{subsec: Physics Approach Misner}}

Kip Thorne\textquoteright s interpretation of Misner space, inspired
by extensive discussions with Charles Misner, offers a distinct and
highly engaging perspective, as also presented in~\cite{Thorne}.
This subsection builds on Thorne's insights to further explore the
nature of Misner space. 

~

Consider the flat $2$-dimensional Minkowski spacetime with a reference
frame in Lorentzian coordinates $(t,x)$, where $t$ is interpreted
as the time direction and $x$ as the space direction. Next, we identify
the $t$-axis at $x=0$ with the $t$-axis at $x=P$ (which we will
henceforth refer to as the $\bar{t}$-axis): $(t,P)\sim(t,0)$. At
$t=0$, we set the $\bar{t}$-axis into motion with constant speed
$\beta$ towards the $t$-axis, while the reference frame in Lorentzian
coordinates remains at rest with respect to the $t$-axis. We can
think of the $t$-axis as representing an observer at rest and the
$\bar{t}$-axis as representing an observer moving at speed $-\beta$
relative the reference frame. This moving observer is equipped with
a clock, and $\tau$ denotes the proper time measured by this clock.\footnote{\emph{Proper time}, also called clock time or process time, is a measure
of the amount of physical process that a system undergoes.} This $x,t$-coordinate system is furnished with the flat Minkowski
metric $\widetilde{\eta}$, and therefore the light cones do not tilt
or open. Due to the special relativistic time dilation, a region with
closed timelike curves develops, see Figure~\ref{fig:Misner moving wall 1}.

\begin{figure}[H]
\centering{}\includegraphics[scale=0.3]{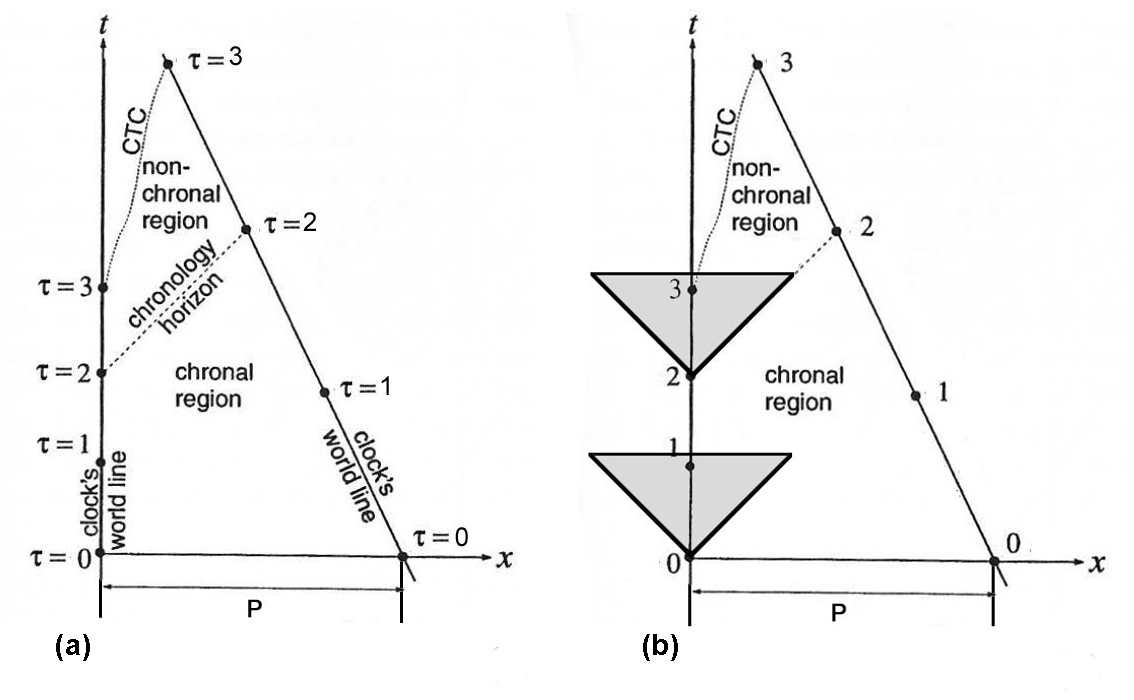}\caption{{\small{}\label{fig:Misner moving wall 1}Depicted are two spacetime
diagrams for Misner space, which illustrate the violation of causality. 
The symbol $\tau$ denotes proper time, and it is important
to note that identical values of proper time on both axes are physically
identical. CTCs appear due to the time dilation effects of special
relativity. The light cone structure is shown in Figure (b). Through
every event to the future of the chronology horizon, there exist CTCs,
whereas there are no CTCs to the past of the chronology horizon. 
}}
\end{figure}

~\\ Note that the worldline of the clock at rest (along the $t$-axis)
and the worldline of the clock in motion (along the $\bar{t}$-axis)
are identified by $\sim$, meaning the events along these lines are
physically identical. However, due to special relativity, the observer
at rest perceives the moving clock to tick more slowly. On the $t$-axis
we have $\tau=t$, while on the $\bar{t}$-axis, due to time dilation,
$\tau=\frac{t}{\sqrt{1-\beta^{2}}}=\gamma t$, where $\gamma:=\frac{1}{\sqrt{1-\frac{v^{2}}{c^{2}}}}=\frac{1}{\sqrt{1-\beta^{2}}}$.

\subsection*{Connection to mathematical approach}

In the style of~\cite{Kim and Thorne}, we set the proper time origins,
$\tau=0$, such that they are separated by the chronology horizon
at $x=t$. We denote the spatial distance between these $\tau=0$
origins as $D$. 

~

The left wall ($t=\tau$, $x=0$) in this example is at rest in Lorentzian
coordinates, while the right wall ($x=D-\beta\gamma\tau$, $t=D+\gamma\tau$)
moves towards the left one with constant velocity $\beta$. The location
of the chronology horizon depends on $D$, which can be any positive
constant, and the speed $\beta$ with $\mid\beta\mid<1$. All further
boosted copies of Misner space depend on this initial choice of $D$
and $\beta$. Provided $\beta\neq0$, the two walls intersect at a
the point $s$, see Figure~\ref{fig:Moving-Wall-Misner D}. This
point, as well as the region beyond it, raises questions about their
nature. We can better address these questions with the aid of the
covering space, as depicted in Figure~\ref{fig:Covering-space-for Misner }.

\begin{figure}[H]
\centering{}\includegraphics[scale=0.29]{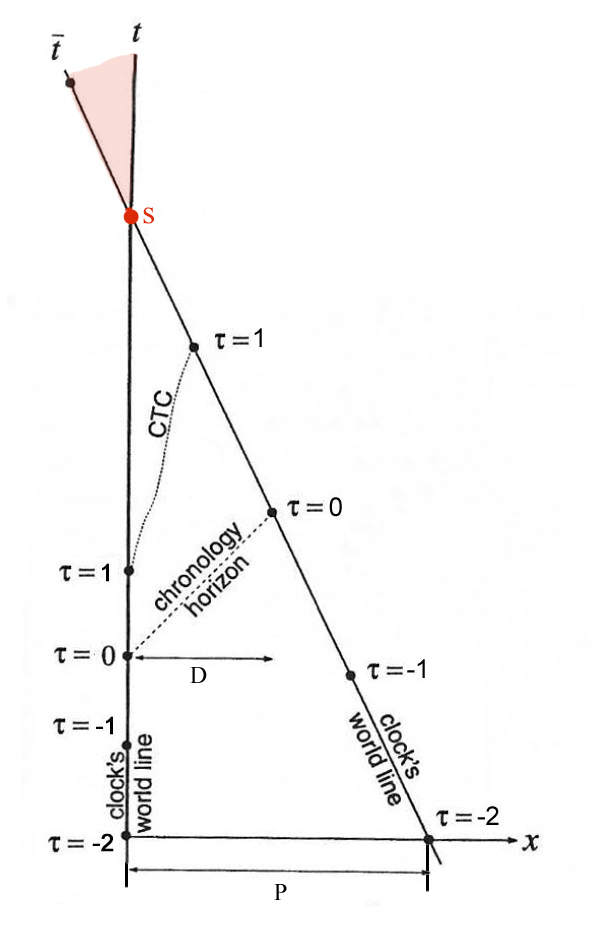}\caption{{\small{}\label{fig:Moving-Wall-Misner D}``Moving wall''
Misner space with proper time origins, such that the two walls are
separated by the chronology horizon at $x=t$ by the spatial distance
$D$.}}
\end{figure}

\subsection*{Calculation of intersection points }

We can describe the \textquoteleft moving wall\textquoteright{} model
of Misner space by switching to the covering space~\cite{Thorne}.
Side by side, we line up copies of Misner space that are identified
through the action of a boost with speed $\beta$. In comparison to
Figure~\ref{fig:Moving-Wall-Misner D}, $W_{0}$ corresponds to $t$,
and $W_{1}$ corresponds to $\bar{t}$. Additionally, the boost induces
an equivalence relation, identifying boost-related points $p$ under
$\sim$ , all of which lie on hyperbolas defined by $\tilde{t}^{2}-\tilde{x}^{2}=\textrm{const}$.

~\\ The relation between the covering space in Minkowski coordinates $(\tilde{t},\,\tilde{x})$
and the Minkowski coordinates $(t,\,x)$ in the \textquoteleft moving
wall\textquoteright{} model of Misner space ($copy\,0$) is given
by~\cite{Kim and Thorne} 

\begin{equation}
\tilde{t}=t-\frac{D}{1-\xi}
\end{equation}

and 
\begin{equation}
\tilde{x}=x-\frac{D}{1-\xi},
\end{equation}

~\\ where $\xi=\sqrt{\left(\frac{1-\beta}{1+\beta}\right)}$ is the
Doppler blueshift. 

\begin{figure}[H]
\centering{}\includegraphics[scale=0.7]{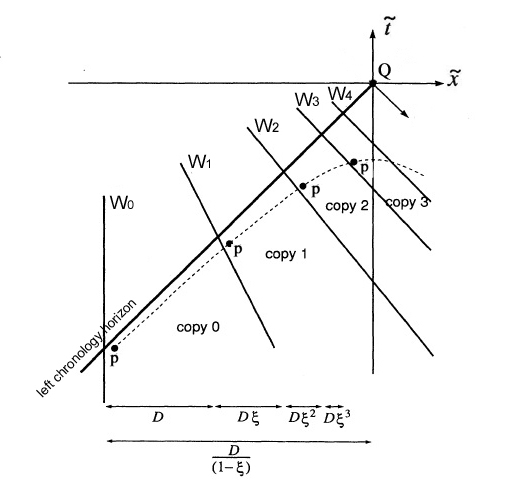}\caption{{\small{}\label{fig:Covering-space-for Misner }Covering space for
the \textquoteleft moving wall\textquoteright{} model of Misner space.
The singularity manifests as the intersection point $Q$ of the left
and right chronology horizons.}}
\end{figure}

~\\ In the covering space (Figure~\ref{fig:Covering-space-for Misner })
we get the spatial distance between $W_{0}$ and the $\tilde{t}$-axis
by the converging series

\begin{gather*}
D\xi^{0}+D\xi^{1}+D\xi^{2}+D\xi^{3}...=D\cdot\sum\xi^{i}=D\cdot\frac{1}{1-\xi},
\end{gather*}

~\\ for $\mid\xi\mid<1$. Thus, the left wall of $copy\,0$, denoted
$W_{0},$ is located at $\tilde{x}=\frac{-D}{1-\xi}$, and the right
wall, $W_{1}$, of $copy\,0$ is located at

\begin{gather*}
(\tilde{t}+\frac{\xi D}{1-\xi})=(\frac{1+\xi^{-2}}{1-\xi^{-2}})(\tilde{x}+\frac{\xi D}{1-\xi}).
\end{gather*}
Hence, 
\begin{equation}
(\frac{1+\xi^{-2(n+1)}}{1-\xi^{-2(n+1)}})(\tilde{x}+\frac{\xi^{n+1}D}{1-\xi})-\frac{\xi^{n+1}D}{1-\xi}=(\frac{1+\xi^{-2n}}{1-\xi^{-2n}})(\tilde{x}+\frac{\xi^{n}D}{1-\xi})-\frac{\xi^{n}D}{1-\xi}
\end{equation}

~\\ yields the intersection of two adjacent walls, $W_{n}\cap W_{n+1}=s_{n}=(\tilde{t}_{n},\tilde{x}_{n})$:

\begin{equation}
\tilde{t}_{n}=\frac{D\xi^{-n}(\xi^{1+2n}-1)}{\xi^{2}-1}
\end{equation}
\begin{equation}
\tilde{x}_{n}=\frac{D\xi^{-n}(\xi^{2n+1}+1)}{\xi^{2}-1}.
\end{equation}

Each point $s_{n}$ is on the hyperboloid 
\begin{equation}
\tilde{t}_{n}^{2}-\tilde{x}_{n}^{2}=-\frac{4\xi D^{2}}{(\xi^{2}-1)^{2}}=\textrm{const}.\label{eq:hyperboloid}
\end{equation}

Clearly, we have $\tilde{t}_{n+1}^{2}-\tilde{x}_{n+1}^{2}=\tilde{t}_{n}^{2}-\tilde{x}_{n}^{2}$
so the points $s_{n+1}$ and $s_{n}$ lie on the same hyperboloid.
From this, we can conclude that each point $s_{n}$ is taken by a
boost to $s_{n+1}$. Thus, the points $s_{n}$ belong to the same
fiber over a specific point in the Misner base space and can be regarded
as physically identical, as they are mapped by the covering map onto
identical events in Misner space. 

~

We turn again to the covering space (see Figure~\ref{Covering Misner - Intersection p}),
we focus on the intersection points and the regions beyond them. At
first glance, the seemingly inconsistent distribution of intersection
points $s_{n}$ might appear to contradict the definition of the covering
space, which is constructed as a union of disjoint sets (in our case,
a disjoint union of copies of Misner space). However, this non-chronological
structure simply reflects the chronology-violating nature of region
$III$, as explained in~\cite{Hawking+Ellis}.

\begin{figure}[H]
\centering{}\includegraphics[scale=0.29]{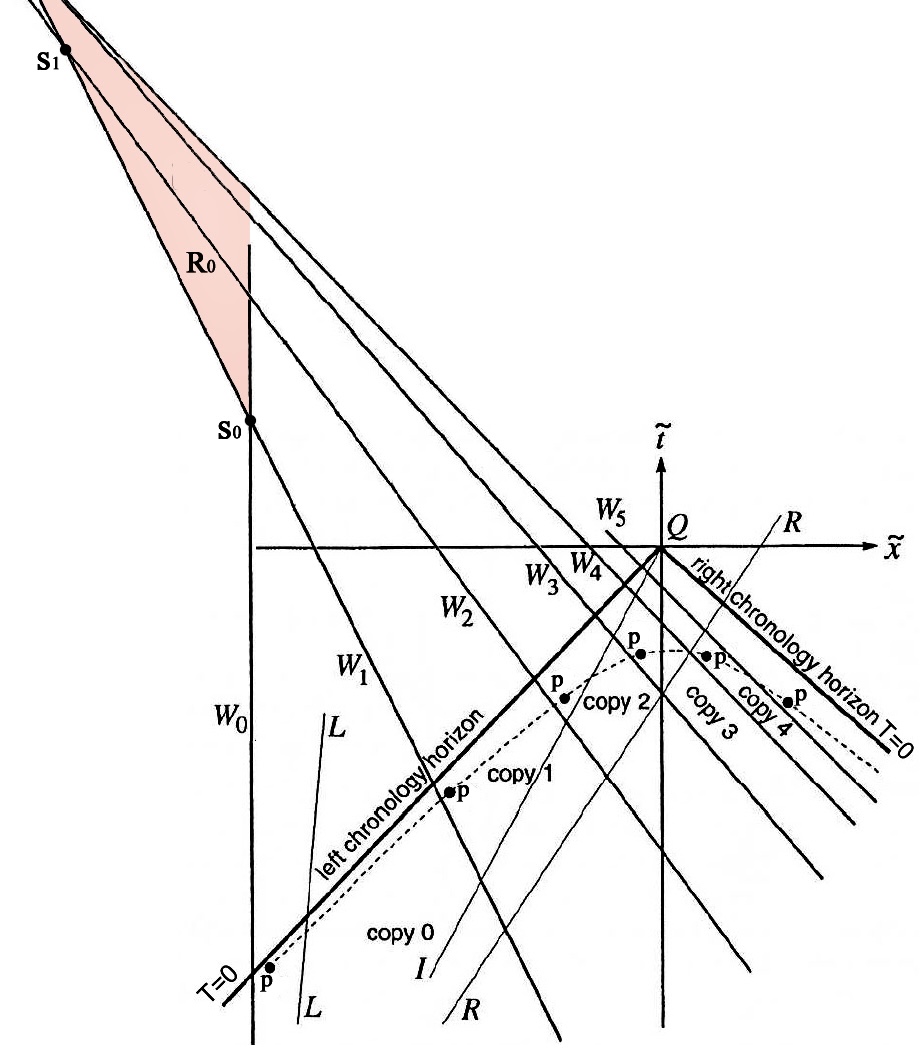}\caption{{\small{}\label{Covering Misner - Intersection p}The covering space
for \textquoteleft moving wall\textquoteright{} model of Misner space
illustrates the intersection points $s_{n}$ and the corresponding
regions $R_{n}$ which are mirror-inverted to the respective region
below $s_{n}$.}}
\end{figure}

~

The above considerations, along with the fact that our initial choice
of $D$ was arbitrary, lead to the conclusion that the intersection
points depend specifically on the choice of the parameter $D$. Consequently,
we may denote these intersection events as $s_{n}(D)$. If the parameter
$D$ is changed, $s_{n}(D)$ moves along the parabola given by (see
Equation~(\ref{eq:hyperboloid}))

\begin{gather*}
f(D)=-\frac{4\xi}{(\xi^{2}-1)^{2}}\cdot D^{2}.
\end{gather*}
To summarize, this ensures that the intersection points $s_{n}$ are
not physical singularities, but merely coordinate singularities. This
reasoning is supported by the fact that the $2$-dimensional Misner
universe can be viewed as a simplified $2D$ wormhole spacetime~\cite{Kim and Thorne};
both spacetimes share the same covering space. By projecting a wormhole
whose mouths move past each other at speed $\beta$~\cite[Figure 3c]{Thorne CTC}
onto the $x-t$ plane, we obtain a situation similar to that of Misner
space, as described above. However, this wormhole spacetime diagram
reveals a new feature: a past chronology horizon, $\mathcal{H}_{-}$.
Therefore, the chronology-violating region is confined by two chronology
horizons. Like the future chronology horizon $\mathcal{H}_{+}$, the
past chronology horizon $\mathcal{H}_{-}$ is a null surface.

~

Given that the initial choice of the wall $W_{0}$ and the speed $\beta$
are arbitrary, we can obtain the covering space without coordinate
singularities~\cite{Ori Misner space} by appropriately selecting
these initial conditions (as shown in Figure~\ref{fig:Ori Covering-space-for Misner }).
The covering space corresponds to the left half-plane of Minkowski
space.

\begin{figure}[H]
\centering{}\includegraphics[scale=1.2]{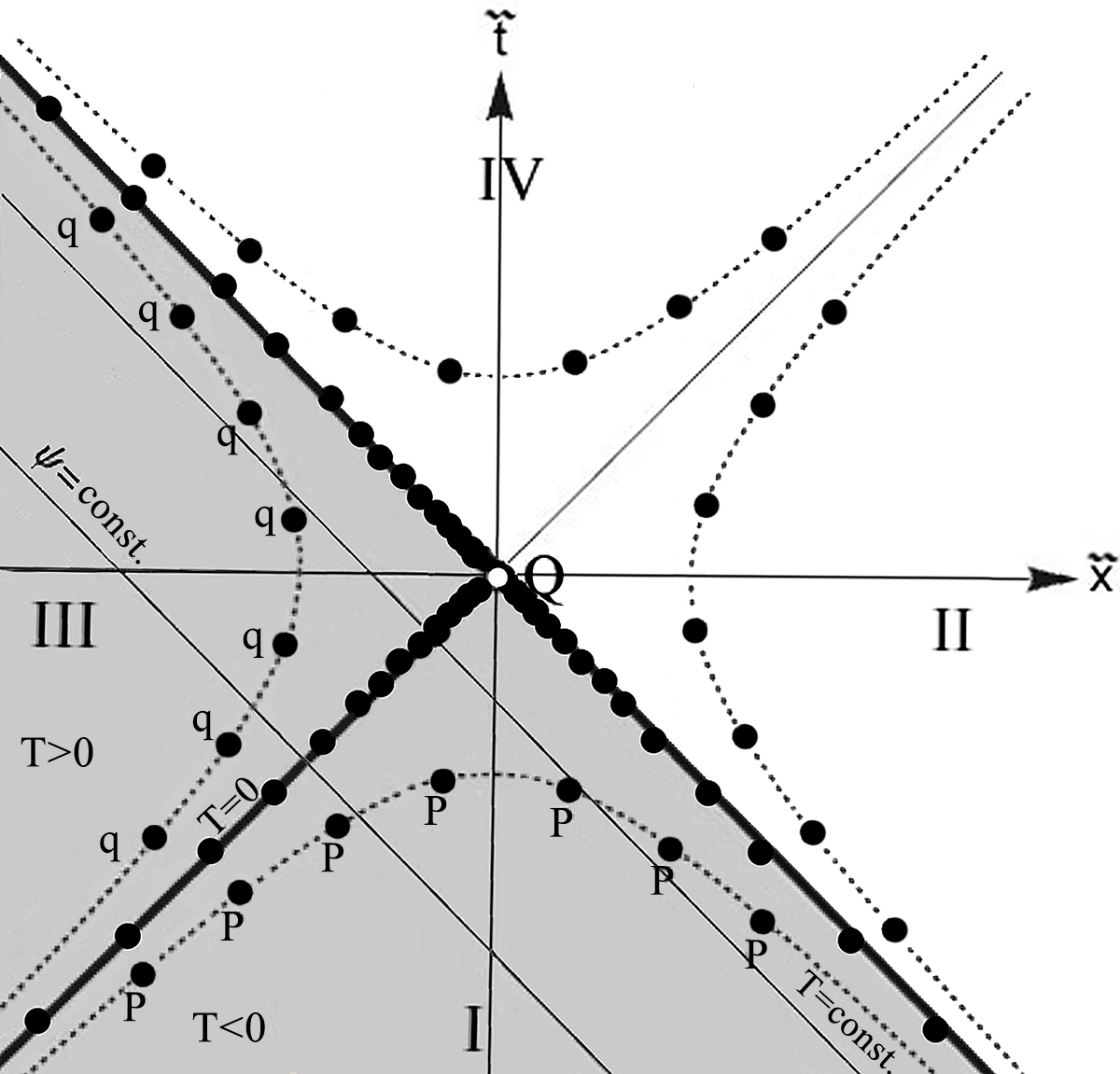}\caption{{\small{}\label{fig:Ori Covering-space-for Misner }Universal covering
space for Misner space with the metric $ds^{2}=Td\psi^{2}+2dTd\psi$.
The transformation between the covering space coordinates and Misner
coordinates is given by $\psi=-2ln(\frac{\tilde{t}-\tilde{x}}{2})$,
$T=\frac{\tilde{t}^{2}-\tilde{x}^{2}}{4}$.}}
\end{figure}

\subsection{Geodesics for warped Misner product space\label{subsec:Misner-Geodesics+warped product}}

The $4$-dimensional Misner space~(\ref{eq:Misner extension1}) can
be expressed as a warped product space of the form $M=M_{Z}\times_{r}H^{2}$,
where $M_{Z}$ is the $2$-dimensional Misner-type spacetime with
coordinates $(T,\varphi)$, and $H^{2}$ is the $2$-dimensional hyperbolic
plane, described by coordinates $(\theta,\phi)$. The general form
of the metric combines the $2$-dimensional Misner-type geometry,
which encodes the non-trivial topological and causal structure, with
a $2$-hyperbolic plane that introduces hyperbolic symmetry. The metric
can be written as

\[
ds^{2}=f(T)(d\varphi)^{2}+2dTd\varphi+r^{2}(d\theta^{2}+\sinh^{2}\theta d\phi^{2}),
\]
where $f(T)$ describes the time-dependent geometry of the Misner-type
spacetime and $r$ is the warping function. To interpret this, we
consider $\mathbb{R}^{2}$ as being represented topologically by the
hyperbolic plane $H^{2}$. The $4$-dimensional Misner space can then
be described as a warped product $M:=M_{Z}\times_{r}H^{2}$, with
the warping function $r$ being a scalar function $r:M_{Z}\rightarrow\mathbb{R}$,
depending on the coordinates of the base manifold $M_{Z}$. The Misner
space $M_{Z}=S^{1}\times\mathbb{R}$ is equipped with the cylindrical
metric $ds^{2}=2dTd\varphi+Td\varphi{}^{2}$. The fiber $H^{2}$,
with coordinates $(\theta,\phi)$, has the hyperbolic metric $g_{H^{2}}=d\theta^{2}+\sinh^{2}\theta d\phi^{2}$,
scaled by $r^{2}$. Therefore, the metric on the total space $M$
can be written as $ds^{2}=g_{Z}+r^{2}g_{H^{2}}$, which explicitly
expands to 
\[
ds^{2}=2dTd\varphi+Td\varphi{}^{2}+r^{2}(d\theta^{2}+\sinh^{2}\theta d\phi^{2}),
\]
where the warping function $r$ is given by $r=T\exp(\frac{\varphi}{2})-\exp(-\frac{\varphi}{2})$.
If $r$ is constant, the warped product reduces to a direct product
of the base $M_{Z}$ and fiber $H^{2}$.
\begin{rem}
The surface $H_{r}^{2}$, also known as a pseudo-sphere, has constant
negative Gaussian curvature of $-1$, which provides the desired high
symmetry. The upper sheet of the pseudo-sphere is given by the set

\begin{gather*}
H_{r}^{2}=\{(z,x,y)\in\mathbb{R}_{1}^{3}:-z^{2}+x^{2}+y^{2}=-r^{2},z>0\}.
\end{gather*}
\end{rem}

~

Based on the warped geometry of Misner space, we can focus on the
embedded cylinders defined by $\theta=\textrm{const}$ and $\varphi=\textrm{const}$.
By doing so, we restrict our attention to the $2$-dimensional Misner
cylinder $M_{Z}$, which captures the essential features of the Misner-type
geometry. 

~

Here, we examine the geodesic equations, which show that there exist
incomplete null and timelike geodesics that spiral around the cylinder
and cannot be extended. One class of these incomplete timelike geodesics
not only orbits around the chronology horizon but also penetrates
the null surface into negative $T$, where it reaches a turning point,
intersects itself, and then encircles the cylinder below the horizon. 

~\\ Given the $2$-dimensional Misner metric $ds^{2}=Td\varphi{}^{2}+2dTd\varphi$,
null geodesics are described by

\begin{equation}
\begin{cases}
\begin{array}{c}
\psi(T)=\textrm{const}.\;\;\;\\
\psi(T)=\int-\frac{2}{T}dT
\end{array} & \begin{array}{c}
(outgoing)\\
(ingoing)\;\;
\end{array},\end{cases}
\end{equation}

~\\ where outgoing null geodesics are complete, while ingoing ones
circle around the Misner cylinder infinitely often near the Cauchy
horizon. For timelike geodesics, we use the expression

\begin{equation}
\psi(T)=\int\frac{\dot{\psi}}{\dot{T}}dT.
\end{equation}

~\\ Substituting $\dot{\psi}$ and $\dot{T}$, we get: 
\[
\psi(T)=\frac{\frac{1}{T}\left(\xi-\omega\sqrt{\xi^{2}+T}\right)}{\omega\sqrt{\xi^{2}+T}}dT=\int\frac{1}{T}\left(\frac{\omega}{\sqrt{1+\frac{T}{\xi^{2}}}}-1\right)dT,
\]

~\\ with $\xi\equiv\textrm{const}$ and $\omega=\pm1$ for outgoing
and ingoing geodesics, respectively. There are two families of timelike
geodesics and both of them are incomplete. One class of geodesics
does not cross the Cauchy horizon and spirals around the cylinder
just above $T=0$. A geodesic from the other class, however, crosses
the Cauchy horizon, entering the chronology-violating region. It reaches
a turning point at $T=-\xi^{2}$, where the particle\textquoteright s
velocity vanishes, as shown by the condition
\begin{gather*}
\omega\sqrt{\xi^{2}+T}=0\Leftrightarrow\sqrt{\xi^{2}+T}=0\Leftrightarrow T=-\xi^{2}.
\end{gather*}

~

Therefore, a massive particle cannot reach $T=-\infty$. Such a timelike
geodesic intersects itself as it tracks back, converging towards the
Cauchy horizon and spiraling around the Misner cylinder just below
$T=0$. In contrast, there is no such behavior for photons, and there
is no obstruction for null geodesics to reach $T=-\infty$.

~

This result implies that geodesics encircle the extended Misner cylinder
with increasingly higher frequencies as they approach $T=0$, asymptotically
nearing the chronology horizon. The singular point $Q$ represents
a genuine spacetime pathology: geodesic incompleteness in this case
could correspond to a rocket ship suddenly disappearing from the universe
after a finite amount of proper time. As shown in Subsection~\ref{subsec: Physics Approach Misner},
this pathological nature of the geodesics becomes evident in the covering
space, where the geodesic singularity manifests as the intersection
point $Q$ of the left and right chronology horizons~\cite{Thorne},
as illustrated in Figure~\ref{Covering Misner - Intersection p}.

\section{General equations for radial geodesics\textbf{ \label{sec:General-Equation-forGeodesics}}}

In this section, we present a complete canonical formulation of the
geodesic equations for cylindrical Misner-type spacetime models, focusing
on the embedded $\theta=\textrm{const}$ and $\varphi=\textrm{const}$
cylinders, as introduced in Subsection~\ref{subsec:Misner-Geodesics+warped product}.
We anticipate that the radial geodesic equations for the pseudo-Schwarzschild
and pseudo-Reissner-Nordstr\"{o}m spacetimes (to be discussed in
Sections~\ref{sec:Pseudo-Schwarzschild-Spacetime} and~\ref{sec:Pseudo-Reissner-Nordstrom-Spacetime})
will exhibit structural similarities. To unify these cases, we derive
a general form of the geodesic equations, which proves to be valid
for all three hyperbolically symmetric spacetimes. Notably, this general
form can also be applied to Misner space, recognizing that the $2$-dimensional
Misner metric, $ds^{2}=Td\varphi{}^{2}+2dTd\varphi$, can be rewritten
in the canonical form 

\begin{equation}
ds^{2}=f(r)d\nu^{2}+2d\nu dr,\label{eq:canonical metric}
\end{equation}

~\\ where $f(r)=g_{\nu\nu}$. This metric, due to its geodesic singularity
at $Q$, is commonly referred to as the canonical quasi-singular metric~\cite{Konkowski}. 

~

To generalize the Equation~\ref{eq:canonical metric} to the $4$-dimensional
case for hyperbolically symmetric models, we start with the general
form of a hyperbolically symmetric metric given by

\begin{equation}
ds^{2}=g_{\nu\nu}d\nu^{2}+2g_{\nu r}d\nu dr+r^{2}(d\theta^{2}+\sinh^{2}\theta d\phi^{2}).\label{eq: general metric form}
\end{equation}
Since we are strictly interested in geodesics for radial motion,
we can safely set $\phi=0$ and $\theta=\frac{\pi}{2}$. Furthermore,
we impose the metric condition resulting from the causal character
of the geodesics

\begin{equation}
g_{\nu\nu}\dot{\nu}^{2}+2g_{\nu r}\dot{r}\dot{\nu}\equiv k,\label{eq:geodesic 1}
\end{equation}

~\\ where $k=-1$ for timelike, $k=+1$ for spacelike, and $k=0$
for null geodesics. The time coordinate $\nu$ is cyclic, and we always
have $g_{\nu r}=1$, so we obtain the constant term

\begin{gather*}
\frac{\partial(\frac{1}{2}g_{\nu\nu}\dot{\nu}^{2}+\dot{r}\dot{\nu})}{\partial\dot{\nu}}=g_{\nu\nu}\dot{\nu}+\dot{r}\equiv\xi.
\end{gather*}

~\\ This is equivalent to $\dot{\nu}=g^{\nu\nu}(\xi-\dot{r})$, and
substituting this into Equation~(\ref{eq:geodesic 1}) gives 
\[
g_{\nu\nu}[g^{\nu\nu}(\xi-\dot{r})]^{2}+2\dot{r[}g^{\nu\nu}(\xi-\dot{r})]=k\Longleftrightarrow\dot{r}^{2}=\xi^{2}-kg_{\nu\nu}.
\]
For a circular orbit, the geodesic equation~\cite{Sharan} is also
given by $\ddot{r}=-\frac{mk}{r^{2}}$.

~\\ Assembling these results, we obtain the equations

\begin{equation}
\dot{\nu}=g^{\nu\nu}\left(\xi-\omega\sqrt{\xi^{2}-kg_{\nu\nu}}\right)
\end{equation}

\begin{equation}
\dot{r}=\omega\sqrt{\xi^{2}-kg_{\nu\nu}},
\end{equation}

~\\ where $\omega$ has been chosen to represent outgoing geodesics
for $\omega=+1$, and ingoing geodesics for $\omega=-1$. Now, the
function of $r$ that is suitable for our study of geodesics can be
inferred:

\begin{equation}
\nu(r)=\int\frac{\dot{\nu}}{\dot{r}}dr=\int\frac{g^{\nu\nu}\left(\xi-\omega\sqrt{\xi^{2}-kg_{\nu\nu}}\right)}{\omega\sqrt{\xi^{2}-kg_{\nu\nu}}}dr=\int\frac{1}{g_{\nu\nu}}\left(\frac{\omega}{\sqrt{1-\frac{kg_{\nu\nu}}{\xi^{2}}}}-1\right)dr.
\end{equation}

~\\ In the case of timelike geodesics this gives 

\begin{equation}
\nu(r)=\int\frac{g^{\nu\nu}\left(\xi-\omega\sqrt{\xi^{2}+g_{\nu\nu}}\right)}{\omega\sqrt{\xi^{2}+g_{\nu\nu}}}dr=\int\frac{1}{g_{\nu\nu}}\left(\frac{\omega}{\sqrt{1+\frac{g_{\nu\nu}}{\xi^{2}}}}-1\right)dr,\label{eq:timelike geodesic}
\end{equation}

~\\ and for null geodesics we have

\begin{equation}
\nu(r)=\int\frac{1}{g_{\nu\nu}}\left(\omega-1\right)dr.\label{eq:null geodesic}
\end{equation}

~\\ For timelike geodesics, the condition $\omega\sqrt{\xi^{2}+g_{\nu\nu}}=0\Longleftrightarrow\xi^{2}+g_{\nu\nu}=0$,
with the value $-\xi^{2}=g_{\nu\nu}$, corresponds to the particle
either having a turning point or asymptotically approaching a horizon.

~

These results provide further insights. By renaming the coordinates
as $\nu=\psi$ and $r=T$, we can carry over the findings to the Misner
case, which takes the same general form as the $2$-dimensional cylindrical
pseudo-Schwarzschild~(\ref{eq: 2D pS metric in EF}) and pseudo-Reissner-Nordstr\"{o}m~(\ref{eq:2D p-RN in EF})
spacetimes in Eddington-Finkelstein coordinates:

\begin{equation}
ds^{2}=f(t)dx^{2}+2dxdt.\label{eq:general metric form}
\end{equation}

\section{Pseudo-Schwarzschild spacetime \label{sec:Pseudo-Schwarzschild-Spacetime}}

Although the pseudo-Schwarzschild spacetime is derived from the well-known
Schwarzschild spacetime, its global and causal structure more closely
resembles that of Misner space. The pseudo-Schwarzschild solution
satisfies the Einstein field equations in vacuum and describes a static,
hyperbolically symmetric vacuum spacetime~\cite{Gaudin, Ori}. Furthermore,
its universal covering space asymptotically approaches Minkowski space,
which is known as asymptotic flatness. This serves as another example
of a spacetime that violates the chronology condition and allows for
the existence of closed timelike curves. The geometry of the pseudo-Schwarzschild spacetime admits a warped product structure, with the base manifold $S^{1}\times\mathbb{R}^{+}$ and the hyperbolic plane $H^{2}$ as the fiber. The point set that defines the pseudo-Schwarzschild manifold is 

\begin{equation}
M=S^{1}\times\mathbb{R}^{+}\times H_{r}^{2},\label{eq: pS manifold}
\end{equation}

~\\ where $S^{1}$ is the $1$-sphere, topologically equivalent to
an interval whose endpoints are glued together by identification,
and $H_{r}^{2}$ (when $t$, $r=\textrm{const}$) represents a $2$-dimensional
surface that can be identified with one of the sheets of the two-sheeted
space-like hyperboloid.\footnote{The upper sheet of the hyperboloid can be globally embedded in $3$-dimensional
Minkowski space. We can express the upper sheet hyperboloid using
a parameterization in polar coordinates $(\theta,\phi)$. Specifically,
the map $\mathbb{R}^{+}\times S^{1}\rightarrow H_{r}^{2}$, $(\theta,\phi)\mapsto(r{\normalcolor \cosh(\theta)},r{\normalcolor \sinh(\theta)}{\normalcolor \cos(\varphi)},r{\normalcolor \sinh(\theta)\sin(\phi)})$
describes the embedding of the upper sheet of the hyperboloid in Minkowski
space. The hyperboloid together with the induced Riemannian metric
$g=d\theta\otimes d\theta+\sinh^{2}(\theta)d\phi\otimes d\phi$ is
called \emph{hyperbolic plane}.} Then the pseudo-Schwarzschild
metric can be derived by performing a Wick rotation $\theta\rightarrow i\theta$
and a signature change from the famous Schwarzschild metric as follows:

~\\ $ds^{2}=-(1-\frac{2m}{r})dt\otimes dt+(1-\frac{2m}{r})^{-1}dr\otimes dr+r^{2}(d\theta\otimes d\theta+\sin{}^{2}(\theta)d\phi\otimes d\phi)$
\begin{center}
\begin{tabular}{ll}
 & \tabularnewline
$\overset{\theta\rightarrow i\theta}{\Longrightarrow}$ & $ds^{2}=-(1-\frac{2m}{r})dt^{2}+(1-\frac{2m}{r})^{-1}dr^{2}+r^{2}(d(i\theta)^{2}+\underset{-\sinh^{2}(\theta)}{\underbrace{\sin^{2}(i\theta)}}d\phi^{2})$\tabularnewline
 & \tabularnewline
$\overset{\cdot(-1)}{\Longrightarrow}$ & $ds^{2}=(1-\frac{2m}{r})dt^{2}-(1-\frac{2m}{r})^{-1}dr^{2}+r^{2}(d\theta^{2}+\sinh^{2}(\theta)d\phi^{2})$.\tabularnewline
 & \tabularnewline
\end{tabular}
\par\end{center}

~\\ We rewrite the last equation as

\begin{equation}
ds^{2}=-(\frac{2m}{r}-1)dt^{2}+(\frac{2m}{r}-1)^{-1}dr^{2}+r^{2}(d\theta^{2}+\sinh^{2}(\theta)d\phi^{2}),\label{eq: pS metric original}
\end{equation}

~\\ where $\theta$ takes all positive values, $0\leq\phi<2\pi$,
and \emph{$m\geq0$}, thereby obtaining the pseudo-Schwarzschild metric.
However, in contrast to the Schwarzschild solution, this spacetime
exhibits hyperbolic symmetry, and the coordinate $t$ is assumed
to be periodic, with $0\leq t<a$, where the identification $(t,r,\theta,\phi)\sim(t+na,r,\theta,\phi)$
for $n\in\mathbb{N}$ is applied. Note that, due to the identification
of edge-to-edge points as defined above, there is a change in the
topology of the spacetime. 

~

The periodicity of the coordinate $t$, after the Wick rotation, arises from the fact that the structure of the corresponding isometry group changes from the non-compact additive group $\mathbb{R}$ to the compact group $U(1)$, which has inherently periodic orbits. In general, the real Lie algebra of the isometry group is generated by the associated Killing vector fields; that is, the infinitesimal generators of one-parameter isometry groups. However, Wick rotating a coordinate (i.e., multiplying it by $i$) changes the associated real Lie algebra into a different real form of a common complexified Lie algebra.

~

For example, the real Lie algebras of $\mathbb{R}$ and $U(1)$ both complexify to the complex Lie algebra $\mathbb{C}$, which is the Lie algebra of the complex multiplicative group $\mathbb{C}^* \setminus \{0\}$. Thus, the transition between $\mathbb{R}$ and $U(1)$ reflects a change in real form within the same complexified structure.

~

In the specific construction used above, this transformation arises via a \emph{double Wick rotation}:
\[
\theta \mapsto i\theta, \qquad t \mapsto it.
\]
Under this transformation, the isometry group associated with $\theta$ changes from $U(1)$ to $\mathbb{R}$, while the group associated with $t$ changes conversely from $\mathbb{R}$ to $U(1)$. Consequently, the new imaginary time coordinate becomes periodic due to the compactness of $U(1)$.

~

This periodicity is not merely a formal artifact but carries physical meaning. For instance, imaginary time is also periodic in the Euclidean Schwarzschild solution, where the length of the period corresponds to the inverse of the Hawking temperature of the original black hole. Hence, the periodic nature of imaginary time in the pseudo-Schwarzschild context is consistent with both geometric reasoning and physical interpretation.

~

Since the geometry of the pseudo-Schwarzschild
spacetime can be expressed as a warped product, with the cylindrical
base $M_{Z}:=S^{1}\times\mathbb{R}^{+}$ and the hyperbolic plane
$H^{2}$ as the fiber, the metric~\ref{eq: pS metric original} can
be decomposed into two parts. We introduce coordinates $(\nu,r)$ on the cylindrical base $M_{Z}$, where $\nu$ serves  as an angular coordinate on $S^{1}$ and $r$ denotes the standard coordinate on $\mathbb{R}^{+}$. The induced cylindrical metric on $M_{Z}$ is then given by $g_{z}:=-(\frac{2m}{r}-1)d\nu^{2}+2d\nu dr$, while the metric on the hyperbolic plane $H^{2}$ is  $g_{H^{2}}:=d\theta^{2}+\sinh^{2}(\theta)d\phi^{2}$, where $(\theta,\phi)$ are pseudo-spherical polar coordinates.

\begin{rem*}
Considering the universal covering space of the pseudo-Schwarzschild
spacetime, we progressively recover flat space as $r\rightarrow\infty$.
The metric, as given in Equation~(\ref{eq: pS metric original}),
exhibits a pathological behavior at $r=2m$ and $r=0$. As we will
show later, $r=2m$ corresponds to a Cauchy horizon, while $r=0$
is a physical singularity analogous to the Schwarzschild singularity.
However, there is a key difference: the pseudo-Schwarzschild singularity
is timelike, whereas the Schwarzschild singularity is spacelike. This
implies that no event inside the horizon can influence any event at
$r>2m$. A simple calculation reveals that $r=0$ represents a curvature
singularity, whereas the point $r=2m$ is merely a coordinate singularity.
Later, we will introduce appropriate coordinates that eliminate the
coordinate singularity.
\end{rem*}
\begin{flushleft}
~
\par\end{flushleft}

The pseudo-Schwarzschild spacetime is static inside the horizon and
time-dependent outside of the horizon. For $r<2m$, the $t$-direction
$\partial_{t}$ is timelike, and $r$ is a spatial coordinate and
the $r$-direction $\partial_{r}$ is spacelike, as seen from the
metric components:

\begin{gather*}
g_{tt}=-(\frac{2m}{r}-1)<0
\end{gather*}

\begin{gather*}
g_{rr}=(\frac{2m}{r}-1)^{-1}>0.
\end{gather*}

~

However, beyond the horizon, $r>2m$, there is a reversal of the roles:
$t$ becomes a spatial coordinate and $r$ becomes a timelike coordinate.
In this region, the decrease of $r$ signifies the passage of time.
As long as you remain in the region $r>2m$, you are inevitably moving
forward in time and will hit the horizon at $r=2m$. All trajectories
lead inevitably to the horizon. Once you cross the horizon, the roles
of $t$ and $r$ are effectively switched, with $t$ becoming a timelike
coordinate and $r$ a spatial one.\footnote{We can analyze the behavior of radial null curves by setting $d\theta=d\phi=0$
and studying the equation $dt=\pm\frac{1}{(2m/r)-1}dr$. For large
$r$, the slope is $\frac{dt}{dr}=\pm1$, which is consistent with
the behavior in flat space. As we approach $r=2m$, we get $\frac{dt}{dr}=\pm\infty$
meaning that the light cones close up at this point. Further, as we
approach the singularity at $r=0$, the light cones are stretched
again, but are tilted according to their reversed role of $t$ and
$r$ as timelike and spacelike coordinates.} Analogous to the Schwarzschild metric, the pseudo-Schwarzschild metric
has a horizon at $r=2m$ and a singularity occurs at $r=0$. As discussed
earlier, the issue with our current coordinates becomes apparent along
radial null geodesics, where the slope $\frac{dt}{dr}$ diverges as
$r$ converges towards $2m$, i.e. $\frac{dt}{dr}\rightarrow\infty$. 

~

To eliminate the coordinate singularity at $r=2m$, we can convert
the pseudo-Schwarzschild coordinates into Eddington-Finkelstein coordinates.
Similar to the Schwarzschild case, we set the term $r+2m\ln\mid r-2m\mid$
equal to the tortoise coordinate $r*$. Then, we set $\nu=-(t+r*),$
which leads to the inverse relationship $t=-(\nu+r+2m\ln\mid r-2m\mid).$
Next, we take the derivative of this equation $dt=-d(\nu+r+2m\ln\mid r-2m\mid),$
and square both sides. By substituting this expression for $dt^{2}$
into the pseudo-Schwarzschild metric, we obtain the following form
of the metric in terms of the tortoise coordinate

\begin{equation}
ds^{2}=-(\frac{2m}{r}-1)d\nu^{2}+2d\nu dr+r^{2}(d\theta^{2}+\sinh^{2}(\theta)d\phi^{2}).\label{eq: pS metric in EF}
\end{equation}
The new coordinate $\nu$ is periodic, with $0\leq\nu<\beta$, similar
to the previous coordinate $t$. A given point $(0,r,\theta,\phi)$
is identified with $(b,r,\theta,\phi)$, where $b$ is the periodicity
of the coordinate $\nu$. These coordinates are naturally adapted
to the null geodesics. Although the metric component $g_{\nu\nu}$
disappears at $r=2m$ , the coordinate transformation eliminates the
singularity that existed at $r=2m$ in the original pseudo-Schwarzschild
coordinates. Moreover, the null hypersurface at $r=2m$ serves as
a Cauchy horizon, denoted $H$, and divides the spacetime $M$ into
two regions: $M_{2}$, where $r<2m$, and $M_{1}$, where $r>2m$.
We can time-orient the manifold $M$ by requiring that $-\partial_{r}$
be future-pointing. This leads to the conclusion that $\partial_{r}$
is a lightlike vector throughout, $\partial_{\nu}$ is timelike for
$r<2m$, and spacelike for $r>2m$.

~ 

Recalling that the pseudo-Schwarzschild spacetime is $M=M_{Z}\times_{r}H^{2}$,
and its metric can be decomposed into the cylindrical metric $g_{Z}$
and the hyperbolic metric $g_{H^{2}}$, we proceed by setting $\theta=\textrm{const}$
and $\phi=\textrm{const}$. Allowing $\nu$ and $r$ to vary and take
the values as previously defined, we obtain the $2$-dimensional cylindrical
pseudo-Schwarzschild spacetime $M_{Z}=S^{1}\times\mathbb{R}^{+}$.
This spacetime is described by the degenerate metric in Eddington-Finkelstein
coordinates (see Figure~\ref{fig:Pseudo-Schwarzschild-spacetime_EF coordinates}):

\begin{equation}
g_{Z}:=-(\frac{2m}{r}-1)d\nu^{2}+2d\nu dr.\label{eq: 2D pS metric in EF}
\end{equation}

\begin{figure}[H]
\centering{}\includegraphics[scale=0.65]{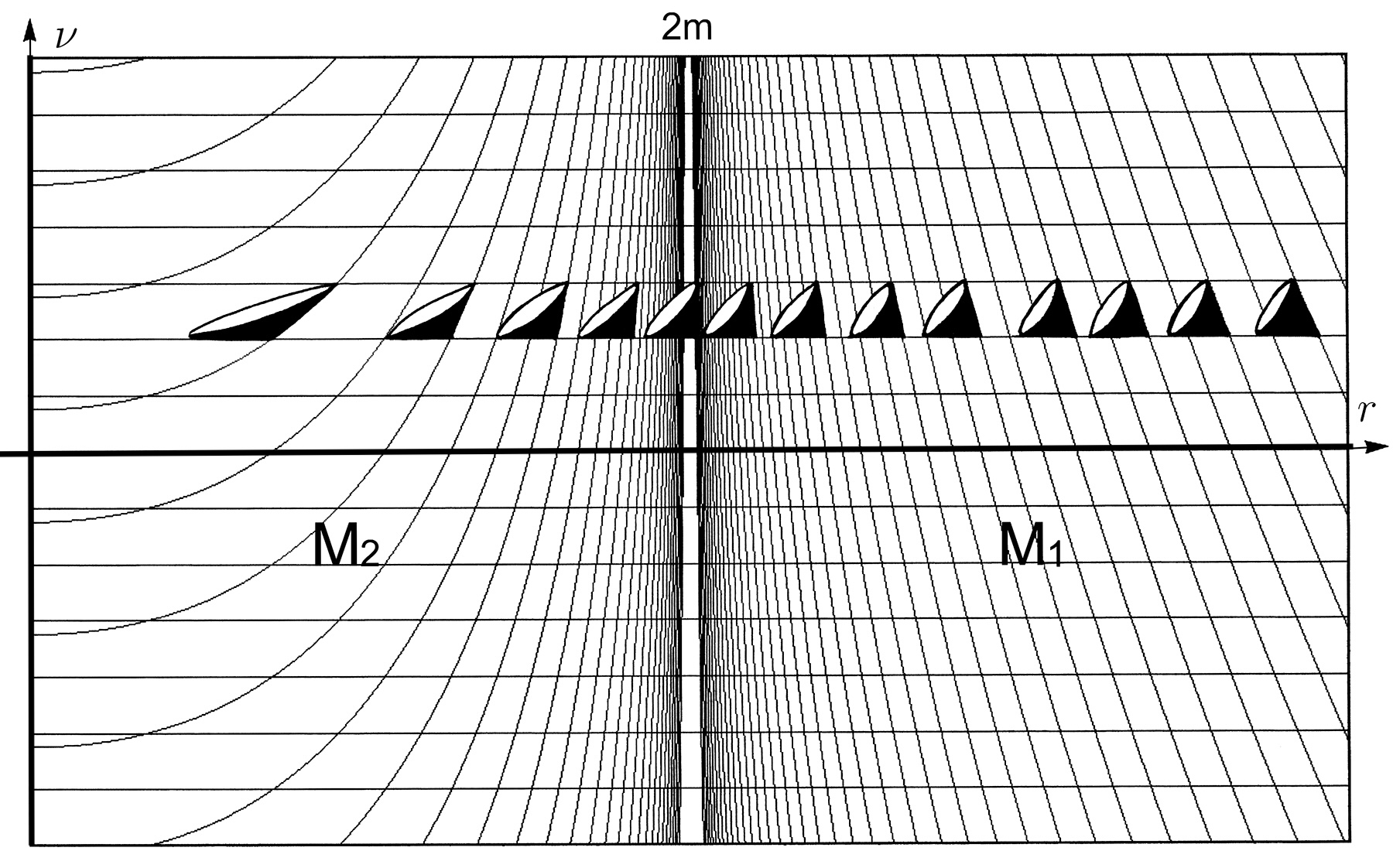}\caption{{\small{}\label{fig:Pseudo-Schwarzschild-spacetime_EF coordinates}Pseudo-Schwarzschild
spacetime in Eddington-Finkelstein coordinates.}}
\end{figure}

~

In two dimensions, the Christoffel symbols are identical for both
the Schwarzschild and pseudo-Schwarzschild spacetimes. Therefore,
the Riemann curvature tensor for the cylindrical pseudo-Schwarzschild
spacetime~(\ref{eq: pS metric in EF}) is given by $-R_{trtr}=-\frac{2m}{r^{3}}\eqqcolon K=\frac{1}{2}S$,
where $K$ denotes the Gaussian curvature and $S$ represents the
scalar curvature. Since we require $m$ to be non-negative and $r$
to be positive, the Gaussian curvature is non-positive throughout
the cylinder.\footnote{The Riemannian curvature tensor for the $4$-dimensional pseudo-Schwarzschild
spacetime can easily computed by the prescription of Wick rotation
and sign change.}

\subsection{Non-chronal region}

In cylindrical coordinates $(\nu,r)$, the $\nu$-coordinate curves
circle around the cylinder. Given the metric metric in Eddington-Finkelstein
coordinates (as presented in Equation~\ref{eq: 2D pS metric in EF}),
besides $\theta,\phi$, we can also set $r$ to a constant value and
consider the one-dimensional line element

\begin{gather*}
g_{\bar{Z}}=-(\frac{2m}{r}-1)d\nu^{2}.
\end{gather*}

The curves $\gamma:[0,b]\rightarrow M$, defined by $\gamma(s)=(s,r,\theta,\phi)$,
are closed because $\gamma(0)=(0,r,\theta,\phi)=(b,r,\theta,\phi)=\gamma(b)$.
Since each curve is contained within a hypersurface of constant $r=\textrm{const}$
and is equipped with the metric $g_{\bar{Z}}$, the nature of the curves
depends on the value of $r$: the curves are timelike if and only
if $r<2m$, the curves are spacelike if $r>2m$. At $r=2m$, the horizon
$H$ is null and consists of closed null curves $\gamma(s)=(s,2m,\theta,\phi)$,
where $\theta$ and $\phi$ are constant. Therefore, $H$ serves as
a chronology horizon for the closed orbits of constant $r,\theta,\phi$.
\begin{prop}
A curve $\gamma:I\rightarrow M_{z}$ with velocity \textup{$\gamma'=\gamma_{1}'\partial_{\nu}+\gamma_{2}'\partial_{r}$}
is causal and future-pointing in $M_{1}$ if and only if $0\leq\frac{2\cdot\gamma_{2}'}{(\frac{2m}{r}-1)}\leq\gamma_{1}$,
and is causal and future-pointing in $M_{2}$ if and only if $\frac{2\cdot\gamma_{2}'}{(\frac{2m}{r}-1)}\geq\gamma_{1}'\geq0$.
\end{prop}

\begin{proof}
We will not go into the details here, as the proof follows a similar
structure to that of Lemma~\ref{lem:p-RN - curve causal + future pointing},
where we discuss the properties of the pseudo-Reissner-Nordstr\"{o}m
spacetime.
\end{proof}
\begin{prop}
There are no closed timelike curves in the region $M_{1}\subset M_{Z}$.
\end{prop}

\begin{proof}
We begin by considering a foliation of $M_{Z}$ using the hypersurfaces
$\mathcal{H}_{r}:=\{(\nu,r)\mid0\leq\nu\leq2\pi\}$ where $r=\textrm{const}$.
The hypersurfaces $\mathcal{H}_{r}$ are spacelike for $r>2m$, thus
we have $M_{1}=\underset{r\in(2m,\infty)}{\bigcup}\mathcal{H}_{r}$.
Given that any curve in this region must be future-pointing, a curve
in $M_{1}$ can only close up if it is contained within a circle in
$\mathcal{H}_{r}$, i.e., if the curve is spacelike and satisfies
the condition $r=\textrm{const}$. This is the only way for the curve
to be periodic and contained within a single hypersurface $\mathcal{H}_{r}$.
Therefore, such curves are spacelike.
\end{proof}
~

The absence of closed timelike curves in region $M_{1}$ implicates
that all closed timelike curves must be part of region $M_{2}$.

\subsection{Pseudo-Schwarzschild geodesics\textmd{\label{subsec: PS Geodesics}}}

The timelike geodesics for the pseudo-Schwarzschild spacetime in Eddington-Finkelstein
coordinates are given by Equation~(\ref{eq:timelike geodesic}),
using the relation $g_{\nu\nu}=-(\frac{2m}{r}-1)$:

\begin{equation}
\nu(r)=\int\frac{-\frac{1}{(\frac{2m}{r}-1)}\left(\xi-\omega\sqrt{\xi^{2}-(\frac{2m}{r}-1)}\right)}{\omega\sqrt{\xi^{2}-(\frac{2m}{r}-1)}}dr=\int\frac{-1}{(\frac{2m}{r}-1)}\left(\frac{\omega}{\sqrt{1-\frac{(\frac{2m}{r}-1)}{\xi^{2}}}}-1\right)dr.\label{eq:pseudo-schwarzschild geodesic}
\end{equation}

~\\ There are two distinct sets of timelike geodesics in the pseudo-Schwarzschild
spacetime. The first set corresponds to a timelike particle that penetrates
the chronology horizon at $r=2m$. This particle has a turning point
at 

\begin{gather*}
r_{turn}=\frac{2m}{(\xi^{2}+1)},
\end{gather*}
~\\ which lies between the singularity and the chronology horizon.
As a result, such a massive particle cannot reach the singularity.
The second set of geodesics corresponds to a particle that spirals
around the horizon as it approaches $r=2m$, but it never crosses
the horizon within this coordinate patch. For photons, the paths are
determined by the equation 

\begin{gather*}
\nu(r)=\int\frac{-1}{(\frac{2m}{r}-1)}\left(\omega-1\right)dr.
\end{gather*}

~\\ Curiously, there is no obstruction for photons to fall into the
singularity.

\section{Relation between pseudo-Schwarzschild and Misner spacetime\label{sec:Relation-Between-Pseudo-Schwarzschild Misner}}

We aim to investigate Misner space and the $2$-dimensional pseudo-Schwarzschild
spacetime $(M_{Z},g_{Z})$, focusing on similarities such as global
causal structure, curvature, and geodesics. In addition to these,
the conformal structure of these spacetimes is notably rich. A conformal
transformation can sometimes convert a curved space with a non-zero
Riemann tensor into a flat space with a zero Riemann tensor. However,
depending on the dimension of the spacetime, a less strict concept
of causal relatedness may be required. In this context, we introduce
the concept of isocausality~\cite{Garcia-Parrado}, which refers
to a situation where the causal structure of spacetime may not be
preserved globally under certain transformations, but the notion of
causal relationships within specific regions remains consistent.

A glance at the causal structure of the cylindrical pseudo-Schwarzschild
space and the $2$-dimensional Misner space might suggest that the two
associated metrics, metric~(\ref{eq: 2D pS metric in EF}) and the
$2$-dimensional variant of metric~(\ref{eq:Misner extension1}),
are related. This raises the natural question of whether an isometric
mapping exists between these manifolds. However, it is known that
any isometry of a pseudo-Riemannian manifold preserves its curvature.
Since the Riemann curvature vanishes in Misner space, but does not
vanish in the cylindrical pseudo-Schwarzschild space, no isometry
exists between them.
\begin{rem}
It is worth noting that in the special case $m=0$, the $2$-dimensional
pseudo-Schwarzschild spacetime is defined on $S^{1}\times\mathbb{R}$
and is flat. However, by setting $m=0$, the chronology horizon in
the pseudo-Schwarzschild space disappears, and closed timelike curves
(CTCs) no longer exist. This distinction makes these spacetimes causally
different, as the chronology horizon is a significant global feature
that separates the chronology-violating region from the region without
CTCs. In what follows, we will restrict our investigation to the case
where $m>0$.

\end{rem}

\subsection{General case\label{subsec:General-Case}}

For radial trajectories in hyperbolically symmetric spacetimes, the
geometry is effectively $1+1$ dimensional. Therefore, for simplicity
and ease of calculation, we first focus on the $2$-dimensional Misner
space and the cylindrical pseudo-Schwarzschild space, as these lead
to relatively straightforward conformal transformations. Afterward,
we turn our attention to $4$-dimensional spacetimes, although the
results from the $2$-dimensional case cannot be directly generalized.
This gives rise to an alternative perspective on causal equivalence,
as introduced by~\cite{Garcia-Parrado}. This approach allows for
the consideration of a much larger class of causally related spacetimes
compared to the classical concept of conformally related spacetimes. 

\subsubsection{Relation in the two-dimensional case}

Since the metrics under study are $2$-dimensional, we know that both
metrics in question belong to the same conformal equivalence class
$[g]$ of conformally flat metrics.\footnote{A metric is said to be conformally flat if it can be conformally transformed into a metric with vanishing Riemann curvature.} Hence, if the metrics $g_{P},g_{M}\in[g]$ are both locally conformal
to a flat metric $g$, i.e. $(g_{M}\sim g)\,\wedge\,(g_{P}\sim g)$,
then by transitivity and symmetry we have $g_{P}\sim g_{M}$. As a
result, we can locally transform the curved pseudo-Schwarzschild space
with a non-vanishing Riemann tensor into the flat Misner space with
a zero Riemann tensor by applying a conformal transformation.

~

Consider the $2$-dimensional flat Misner space defined on $S^{1}\times\mathbb{R}$
and equipped with the metric 

\begin{equation}
ds^{2}=Td\varphi{}^{2}-dTd\varphi.\label{eq: 2D Misner space}
\end{equation}

\begin{rem}
This version of Misner metric can be obtained from the $2$-dimensional
Minkowski metric~(\ref{eq:minkowski metric}) by the initial coordinate
transformation $\tilde{t}=\eta\cosh(X_{1})$ and $\tilde{x}_{1}=\eta\sinh(X_{1})$.
This coordinate change results in the Misner metric $ds^{2}=-d\eta^{2}+\eta^{2}(dX_{1})^{2}$,
where $0<\eta<\infty$, $0\leq X_{1}\leq2\pi$. Another coordinate
change with $\varphi=X_{1}+\ln(\eta)$ and $T=\eta^{2}$ results in
the Misner metric $-dTd\varphi+Td\varphi{}^{2}$, where the domains
are $-\infty<T<\infty$ and $0\leq\varphi\leq2\pi$.

~
\end{rem}

The pseudo-Schwarzschild spacetime is defined on $S^{1}\times(0,\infty)$,
with the corresponding line-element 
\begin{equation}
ds^{2}=-(\frac{2m}{r}-1)d\nu^{2}+2d\nu dr.\label{eq: 2D pseudo-schwarzschild}
\end{equation}
Obviously the base manifolds are diffeomorphic. We start with Equation~(\ref{eq: 2D pseudo-schwarzschild}).
For convenience, we introduce dimensionless quantities $\bar{r}=\frac{r}{m}$,
$\bar{t}=\frac{t}{m}$ and rewrite the metric~((\ref{eq: 2D pseudo-schwarzschild}))
as 

~
\begin{center}
\begin{tabular}{ll}
$ds^{2}$ & $=-(\frac{2}{\bar{r}}-1)d(\bar{\nu}m)^{2}+2d(\bar{\nu}m)d(\bar{r}m)$\tabularnewline
 & \tabularnewline
 & $=m^{2}[-(\frac{2}{\bar{r}}-1)d\bar{\nu}^{2}+2d\bar{\nu}d\bar{r}]$\tabularnewline
 & \tabularnewline
 & $=m^{2}(\frac{2}{\bar{r}}-1)[-d\bar{\nu}^{2}+\underset{d(r^{*})}{\underbrace{2(\frac{2}{\bar{r}}-1)^{-1}d\bar{r}}}d\bar{\nu}]$.\tabularnewline
\end{tabular}
\par\end{center}

~\\ Next we introduce 
\begin{equation}
d(r*)^{2}:=2(\frac{2}{\bar{r}}-1)^{-1}d\bar{r},
\end{equation}
~\\ with $r^{*}=-2(\bar{r}+2\log(\bar{r}-2))$. This implies the
line-element 
\begin{gather*}
ds^{2}=m^{2}(\frac{2}{\bar{r}}-1)[-d\bar{\nu}^{2}+d(r^{*})d\bar{\nu}].
\end{gather*}

~\\ We then redefine both coordinates as $\varphi=\alpha\bar{\nu}$
and $T=\textrm{e}^{\alpha r^{*}}$, where $\alpha$ is an arbitrary
constant. Considering that $r^{*}=\frac{1}{\alpha}\cdot\log(T)$ yields

\begin{gather*}
d\bar{\nu}=d\frac{\varphi}{\alpha}=\frac{1}{\alpha}d\varphi
\end{gather*}
and 
\begin{gather*}
d(r^{*})=d(\frac{1}{\alpha}\cdot\log(T))=\frac{1}{\alpha T}dT.
\end{gather*}

~\\ With these expressions we obtain as line-element

~

\begin{center}
\begin{tabular}{ll}
$ds^{2}$ & $=m^{2}(\frac{2}{\bar{r}}-1)[-(\frac{1}{\alpha}d\varphi)^{2}+(\frac{1}{\alpha T}dT)(\frac{1}{\alpha}d\varphi)]$\tabularnewline
 & \tabularnewline
 & $=m^{2}(\frac{2}{\bar{r}}-1)[-\frac{1}{\alpha^{2}}d\varphi{}^{2}+\frac{1}{\alpha^{2}T}dTd\varphi]$.\tabularnewline
\end{tabular}
\par\end{center}

~\\ By multiplication with $(-1)$ we get

~

\begin{center}
\begin{tabular}{ll}
$ds^{2}$ & $=-m^{2}(\frac{2}{\bar{r}}-1)\frac{1}{\alpha^{2}}\frac{1}{T}[-Td\varphi{}^{2}+dTd\varphi]=\underset{\Omega(T)}{\underbrace{(\frac{2}{\bar{r}}-1)\frac{m^{2}}{T\alpha^{2}}}}[-dTd\varphi+Td\varphi{}^{2}]$\tabularnewline
 & \tabularnewline
 & $=\Omega(T)[-dTd\varphi+Td\varphi{}^{2}]$.\tabularnewline
\end{tabular}
\par\end{center}

~\\ The conformal factor $\Omega$ is singular at the original horizon
$r=2m$. However, to ensure that the conformal factor is non-zero,
we set the following expression for $T$ (based on the defined terms
above):

\begin{gather*}
T=\textrm{e}^{\alpha r*}=\textrm{e}^{\alpha(-2\bar{r}-4\log(\bar{r}-2))}=\textrm{e}^{-2\alpha\bar{r}}\cdot\textrm{e}^{-\alpha4\log(2-\bar{r})}=\textrm{e}^{-2\alpha\bar{r}}\cdot(2-\bar{r})^{-4\alpha}.
\end{gather*}

~\\ Hence,

\begin{center}
\begin{tabular}{ll}
$\Omega(\bar{r})$ & $=(\frac{2}{\bar{r}}-1)\frac{m^{2}}{\alpha^{2}(\textrm{e}^{-2\alpha\bar{r}}\cdot(2-\bar{r})^{-4\alpha})}=\frac{m^{2}}{\bar{r}\alpha^{2}}(2-\bar{r})\frac{1}{(2-\bar{r})^{-4\alpha}}\cdot\textrm{e}^{2\alpha\bar{r}}$\tabularnewline
 & \tabularnewline
 & $=\frac{m^{2}}{\bar{r}\alpha^{2}}(2-\bar{r})(2-\bar{r})^{4\alpha}\cdot\textrm{e}^{2\alpha\bar{r}}=\frac{m^{2}\textrm{e}^{2\alpha\bar{r}}}{\bar{r}\alpha^{2}}(2-\bar{r})^{4\alpha+1}$,\tabularnewline
 & \tabularnewline
\end{tabular}
\par\end{center}

~\\ and we choose $\alpha=-\frac{1}{4}$ to make $\Omega$ regular
on the chronology horizon, e.g. to ensure $\Omega\neq0$. Going back
to the original coordinates, we have 
\begin{gather*}
\Omega(r)=\frac{m^{2}\textrm{e}^{\frac{-r}{2m}}}{\frac{r}{16m}}(2-\frac{r}{m})^{0}=\frac{16m^{3}\textrm{e}^{\frac{-r}{2m}}}{r}
\end{gather*}
for the conformal factor. Furthermore, we rewrite the conformal factor
in terms of the actual coordinates $T,\,\varphi$. Note that $r=\bar{r}m=2m\left(W_{n}(-\frac{T}{2\textrm{e}})+1\right)$,
$n\in\mathbb{Z}$, where the expression $W$ is the \emph{product
log }function. Replacing $r$ by $2m\left(W(-\frac{T}{2\textrm{e}})+1\right)$
yields the conformal factor

\begin{center}
\begin{tabular}{ll}
$\Omega(T)$ & $=\frac{16m^{3}\textrm{e}^{\frac{-r}{2m}}}{r}=\frac{16m^{3}\textrm{e}^{\frac{-2m\left(W(-\frac{T}{2\textrm{e}})+1\right)}{2m}}}{2m\left(W(-\frac{T}{2\textrm{e}})+1\right)}=\frac{8m^{2}\textrm{e}^{-\left(W(-\frac{T}{2\textrm{e}})+1\right)}}{\left(W(-\frac{T}{2\textrm{e}})+1\right)}$\tabularnewline
 & \tabularnewline
 & $=\frac{8m^{2}\textrm{e}^{-\left(W(-\frac{T}{2\textrm{e}})+1\right)}}{W(-\frac{T}{2\textrm{e}})+1}=\frac{8m^{2}}{\textrm{e}^{\left(W(-\frac{T}{2\textrm{e}})+1\right)}\cdot W(-\frac{T}{2\textrm{e}})+1}=\frac{8m^{2}}{\textrm{e}^{\left(W(-\frac{T}{2\textrm{e}})+1\right)}-\frac{T}{2}}$\tabularnewline
 & \tabularnewline
\end{tabular}
\par\end{center}

~\\ and the corresponding line-element 
\begin{gather*}
ds^{2}=\,\underset{\Omega(T)}{\underbrace{\frac{8m^{2}}{\textrm{e}^{\left(W(-\frac{T}{2\textrm{e}})+1\right)}-\frac{T}{2}}}}[-dTd\varphi+Td\varphi{}^{2}],
\end{gather*}
 ~\\ where the conformal factor, $\Omega(T)$, is a smooth, non-zero
scalar function of the spacetime coordinates. Conformal transformations
preserve the local causal structure: they map timelike curves to timelike
curves, spacelike curves to spacelike curves, and null curves to null
curves. Furthermore, we have the well-known
\begin{prop}
Conformal transformations between $2$-dimensional Lorentzian manifolds
send null hypersurfaces to null hypersurfaces and therefore null geodesics
to null geodesics.
\end{prop}

However, timelike geodesics in one spacetime are not necessarily geodesics
in another conformally equivalent spacetime, even if both spacetimes
are related by a conformal diffeomorphism. Nonetheless, the treatments
of timelike geodesics in Equation~\ref{subsec:Misner-Geodesics+warped product}
and Equation~\ref{subsec: PS Geodesics} demonstrate that the geodesic
behavior in both spacetimes is, in fact, similar.

~

This means that for the $1+1$ dimensional pseudo-Schwarzschild and
Misner spacetimes, we can identify conformal transformations and demonstrate
that timelike geodesics are of a similar type in both spacetimes.
This establishes a notable relationship between them, with the $2$-dimensional
pseudo-Schwarzschild spacetime being viewed as a modification or generalization
of Misner space. However, the $4$-dimensional pseudo-Schwarzschild
solution is not conformally flat, as indicated by the non-zero components
of the conformally invariant Weyl tensor. In fact, only one independent
spin component of the Weyl tensor does not vanish, specifically the
Weyl scalar $C=\frac{m}{r^{3}}$. Therefore, in this context, a less
strict definition of causal relatedness is required. 

\subsubsection{Relation in the $4$-dimensional case}

Since the Weyl curvature tensor vanishes identically in two dimensions,
all $2$-dimensional Lorentz metrics belong to the same conformal
class of flat metrics. Therefore, it is relatively straightforward
to find conformal transformations between any two representatives
of this conformal class. In dimensions greater than three, the Weyl
curvature tensor is generally non-zero, and a necessary condition
for a spacetime to be conformally flat is that the Weyl tensor vanishes.
As a result, establishing an explicit conformal relation between $4$-dimensional
spacetimes becomes significantly more challenging.

~

In four dimensions, Misner space is flat, while pseudo-Schwarzschild
spacetime is not conformally flat. Therefore, these two spacetimes
do not belong to the same conformal class. At this point, we turn
to the concept of causal relations, known as isocausality, which is
more general than conformal relations. We consider mappings that preserve
causal relations, even if their inverses do not necessarily do so.
This notion was proposed by Garc\'{i}a-Parrado and Senovilla~\cite{Garcia-Parrado}
and further developed by A. Garc\'{i}a-Parrado and M. S\'{a}nchez~\cite{Garcia-Parrado_M. Sanchez}.
The analogy between the pseudo-Schwarzschild and Misner spacetimes
becomes clearer through Penrose diagrams, as shown in Figure~\ref{fig:Penrose Diagram}.
These diagrams reveal that the universal covering spaces of these
spacetimes share the same large-scale structure.

\begin{figure}[H]
\centering{}\includegraphics[scale=0.33]{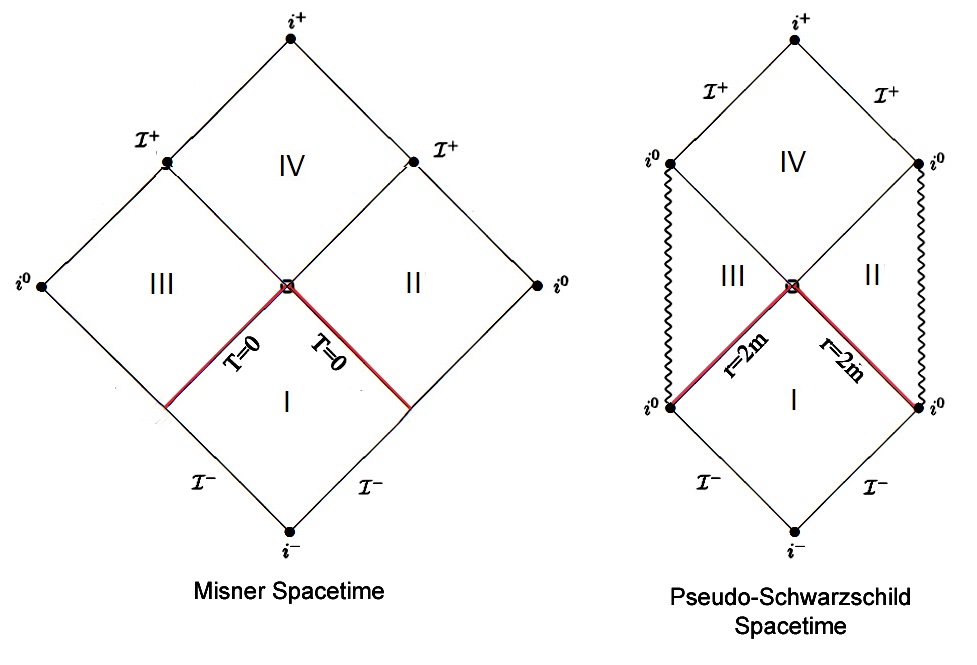}\caption{{\small{}\label{fig:Penrose Diagram}Conformal diagrams for the coverings
spaces of Misner space and pseudo-Schwarzschild spacetime. The resemblance
is noteworthy and we suspect a causal relationship. According to~\cite{Garcia-Parrado},
Penrose diagrams are very helpful to get a clearer picture of the
large-scale structure and to detect isocausality.}}
\end{figure}

\begin{rem}
The boundaries in the diagram are not part of the original spacetimes.
It is important to remember that we imposed periodicity on one coordinate
in the original spacetimes, which cannot be appropriately represented
in the Penrose diagrams.
\end{rem}

In earlier work~\cite{Rieger - Topologies of maximally extended non-Hausdorff Misner Space} 
we conjectured that the $4$-dimensional pseudo-Schwarzschild spacetime and Misner space are isocausal. 
In the present paper, we prove this result and extend it to include the pseudo-Reissner-Nordstr\"{o}m spacetime. 
The proof, given in Proposition~\ref{prop:IsoAllThree}, establishes that all three models are pairwise isocausal 
on their universal covers and on suitable causally regular regions of the compactified spacetimes, 
and provides explicit criteria for when this equivalence descends globally.

\section{Pseudo-Reissner-Nordstr\"{o}m spacetime\label{sec:Pseudo-Reissner-Nordstrom-Spacetime}}

We begin with the well-known Reissner-Nordstr\"{o}m spacetime and
modify it to create a novel spacetime, which we refer to as the pseudo-Reissner-Nordstr\"{o}m
spacetime. The two metrics are related through Wick rotation and a
change in signature. Comprehensive references for the Reissner-Nordstr\"{o}m
spacetime include~\cite{Chandrasekhar,Hawking+Ellis,Misner_Thorne}.
Our focus will be on the geometry of the newly derived pseudo-Reissner-Nordstr\"{o}m
spacetime, examining its global structure, geodesics, and causal properties.
Furthermore, we will explore how this spacetime relates to the pseudo-Schwarzschild
and Misner spacetimes, with many of the properties we derive also
applying to these two spacetimes.

~\\ By performing a Wick rotation $\theta\rightarrow i\theta$ on
the Reissner-Nordstr\"{o}m metric, changing the range of $\theta$
from $[0,\pi]$ to $[0,+\infty]$, flipping the signature, and identifying
$t+P$ with $t$, we obtain the manifold

\begin{gather*}
M:=S^{1}\times(0,\infty)\times H_{r}^{2},
\end{gather*}

~\\ with coordinates $(t,r,\theta,\varphi)$, equipped with the metric
\begin{equation}
ds^{2}=-(\frac{2m}{r}-\frac{q^{2}}{r^{2}}-1)dt^{2}+(\frac{2m}{r}-\frac{q^{2}}{r^{2}}-1)^{-1}dr^{2}+r^{2}(d\theta^{2}+\sinh^{2}(\theta)d\phi^{2}).\label{eq:p-RN original metric}
\end{equation}



~

Similar to the pseudo-Schwarzschild manifold, the $1$-sphere $S^{1}$
(with $0\leq t<t+\alpha$) is topologically equivalent to an interval
with its endpoints identified. The surface $H_{r}^{2}$ is $2$-dimensional
and corresponds to the upper sheet of a two-sheeted spacelike hyperboloid.
The coordinate $\theta$ can take any positive value, while $\phi$
is periodic and ranges from $0\leqslant\phi<2\pi$. The real constants
$m$ and $q$ represent mass and electric charge, respectively, and
are assumed to be positive. This new metric can also be derived by
performing a Wick rotation and a signature change on the Schwarzschild
metric, followed by the substitution

\begin{gather*}
m\rightarrow m(r):=m-\frac{q^{2}}{2r},
\end{gather*}

~\\ which reduces back to the pseudo-Schwarzschild metric when $q=0$.
The metric, as given in Equation~(\ref{sec:Relation-Between-Pseudo-Schwarzschild Misner}),
is singular when 

\begin{gather*}
\frac{2m}{r}-\frac{q^{2}}{r^{2}}-1=0\:\Longleftrightarrow r_{\pm}=m\pm\sqrt{m^{2}-q^{2}},
\end{gather*}

~\\ and also at $r=0$.

~

The singularities at $r=r_{-}$ and $r=r_{+}$ are not true physical
singularities; rather, they represent quasi-regular singularities arising
from geodesic incompleteness.

\begin{rem*}
The Reissner-Nordstr\"{o}m spacetime is not a vacuum solution of Einstein's field equations, as it contains matter in the form of an electromagnetic field. It is instead an \emph{electrovacuum} solution of the Einstein-Maxwell equations, described by an energy-momentum tensor $T_{\mu\nu}$ sourced by the electromagnetic field (from a vector potential). The metric, with signature $(-+++)$, satisfies the Einstein equations in the form $G_{\mu\nu} = R_{\mu\nu} - \frac{1}{2} R g_{\mu\nu} = \kappa T_{\mu\nu}$, where $T_{\mu\nu}$ is the energy-momentum tensor of the electromagnetic field.

~

Our construction of the pseudo-Reissner-Nordstr\"{o}m spacetime begins by performing a Wick rotation on the angular coordinate, $\theta \rightarrow i\theta$, followed by a sign reversal of the metric tensor, $g'_{\mu\nu} = -g_{\mu\nu}$. This transformation results in a hyperbolic geometry in the angular sector. While the Einstein tensor $G_{\mu\nu}$ remains invariant under this sign reversal, the energy-momentum tensor, which is constructed from the metric, changes sign such that $T'_{\mu\nu} = -T_{\mu\nu}$. Consequently, the field equations for the new metric become $G_{\mu\nu} = \kappa T'_{\mu\nu} = -\kappa T_{\mu\nu}$.

~

This resulting spacetime is no longer a solution of the standard Einstein-Maxwell equations. The gravitational source now corresponds to an energy-momentum tensor that violates the weak energy condition, as indicated by $T_{\mu\nu}k^\mu k^\nu < 0$ for some timelike vector $k^\mu$. This violation implies the presence of a source of \emph{exotic matter} with negative energy density. This situation differs significantly from the pseudo-Schwarzschild metric, which remains a valid vacuum solution because its original stress-energy tensor was zero and remains so after the analogous transformation.

~

This distinction is critical. While the Misner and pseudo-Schwarzschild spacetimes are both vacuum solutions, the pseudo-Reissner-Nordstr\"{o}m spacetime is a non-vacuum solution that requires exotic matter. This non-vacuum character is what makes it a particularly surprising and significant member of the family of causality-violating spacetimes we are investigating.
\end{rem*}

Before discussing causality, we will change coordinates and define the extended manifold. There exists a set of coordinates that better describe the singularities at $r=r_{-}$ and $r=r_{+}$.
To achieve this, we introduce the Eddington-Finkelstein coordinates $(\nu,r,\theta,\phi)$,
which can be introduced as follows:

~\\ We regard $\theta$ and $\phi$ being constant and first compute
the tortoise coordinate $r*$ by setting 

\[
ds^{2}=-(\frac{2m}{r}-\frac{q^{2}}{r^{2}}-1)dt^{2}+\frac{1}{-(\frac{2m}{r}-\frac{q^{2}}{r^{2}}-1)}dr^{2}=0.
\]
We can now assert that 
\[
\frac{dt^{2}}{dr^{2}}=\frac{1}{(\frac{2m}{r}-\frac{q^{2}}{r^{2}}-1)^{2}}\Longrightarrow\frac{dt}{dr}=\pm\frac{1}{(\frac{2m}{r}-\frac{q^{2}}{r^{2}}-1)}.
\]

~\\ Then we integrate $\frac{dt}{dr}=\frac{1}{(\frac{2m}{r}-\frac{q^{2}}{r^{2}}-1)}$
to obtain

~\\{}
\begin{center}
\begin{tabular}{ll}
$t(r)$ & $=\int\frac{1}{(\frac{2m}{r}-\frac{q^{2}}{r^{2}}-1)}dr$\tabularnewline
 & \tabularnewline
 & $=(2m^{2}-q^{2})\cdot\tan^{-1}(\frac{r-m}{\sqrt{q^{2}-m^{2}}})-m\cdot\log(-2mr-q^{2}+r^{2})-r+c,$\tabularnewline
 & \tabularnewline
\end{tabular}
\par\end{center}

~\\ which yields the tortoise coordinate which is defined by
\begin{gather*}
r*:=(2m^{2}-q^{2})\cdot\tan^{-1}(\frac{r-m}{\sqrt{q^{2}-m^{2}}})-m\cdot\log(-2mr-q^{2}+r^{2})-r.
\end{gather*}

~

Thus, we obtain the Eddington-Finkelstein coordinates $\nu=t+r*$.
Our next step is to express the pseudo-Reissner-Nordstr\"{o}m metric
in terms of the tortoise coordinate. We define $t=\nu-r*$ and substitute
$t$ with the term $\nu-r*$ in the original metric~(\ref{eq:p-RN original metric}):

~

\begin{center}
\begin{tabular}{ll}
$ds^{2}$ & $=-(\frac{2m}{r}-\frac{q^{2}}{r^{2}}-1)\,(d(\nu-r*))^{2}+(\frac{2m}{r}-\frac{q^{2}}{r^{2}}-1)^{-1}dr^{2}+r^{2}(d\theta^{2}+\sinh^{2}(\theta)d\phi^{2})$\tabularnewline
 & \tabularnewline
 & $=-(\frac{2m}{r}-\frac{q^{2}}{r^{2}}-1)\,(\underset{(\frac{2m}{r}-\frac{q^{2}}{r^{2}}-1)^{-1}}{\underbrace{\frac{\partial(\nu-r*)}{\partial r}}}dr+\underset{1}{\underbrace{\frac{\partial(\nu-r*)}{\partial\nu}}}d\nu)^{2}$\tabularnewline
 & \tabularnewline
 & $+\,(\frac{2m}{r}-\frac{q^{2}}{r^{2}}-1)^{-1}dr^{2}+r^{2}(d\theta^{2}+\sinh^{2}(\theta)d\phi^{2})$\tabularnewline
 & \tabularnewline
 & $=-(\frac{2m}{r}-\frac{q^{2}}{r^{2}}-1)\,[(\frac{1}{(\frac{2m}{r}-\frac{q^{2}}{r^{2}}-1)^{2}})dr^{2}-2\cdot1\cdot\frac{1}{(\frac{2m}{r}-\frac{q^{2}}{r^{2}}-1)}drd\nu+1d\nu^{2}]$\tabularnewline
 & \tabularnewline
 & $+\,(\frac{2m}{r}-\frac{q^{2}}{r^{2}}-1)^{-1}dr^{2}+r^{2}(d\theta^{2}+\sinh^{2}(\theta)d\phi^{2})$\tabularnewline
 & \tabularnewline
 & $=-(\frac{2m}{r}-\frac{q^{2}}{r^{2}}-1)^{-1}dr^{2}+2drd\nu-(\frac{2m}{r}-\frac{q^{2}}{r^{2}}-1)d\nu^{2}$\tabularnewline
 & \tabularnewline
 & $+\,(\frac{2m}{r}-\frac{q^{2}}{r^{2}}-1)^{-1}dr^{2}+r^{2}(d\theta^{2}+\sinh^{2}(\theta)d\phi^{2})$\tabularnewline
 & \tabularnewline
 & $=-(\frac{2m}{r}-\frac{q^{2}}{r^{2}}-1)d\nu^{2}+2drd\nu+r^{2}(d\theta^{2}+\sinh^{2}(\theta)d\phi^{2})$.\tabularnewline
\end{tabular}
\par\end{center}

~\\ Hence, in terms of the Eddington-Finkelstein coordinates, the
metric takes the form

\begin{equation}
ds^{2}=-(\frac{2m}{r}-\frac{q^{2}}{r^{2}}-1)\,d\nu^{2}+2drd\nu+r^{2}(d\theta^{2}+\sinh^{2}(\theta)d\phi^{2}),\label{eq:4D p-RN in EF}
\end{equation}

~\\ where $\nu$ has the same periodicity as $t$. The metric is
now regular at $r=r_{-}$ and $r=r_{+}$, but we still have an irremovable
singularity at $r=0$. Fortunately, the first two terms are independent
of the coordinates $\theta$ and $\phi$. To align the pseudo-Reissner
Nordstr\"{o}m spacetime with the other examples discussed in this
article, and because most of our calculations only apply to the coordinates
$\nu$ and $r$, we can ignore the term $r^{2}(d\theta^{2}+\sinh^{2}(\theta)d\phi^{2})$.
From now on we treat $\theta$ and $\phi$ as constants and consider
the cylindrical metric

\begin{equation}
ds^{2}=-(\frac{2m}{r}-\frac{q^{2}}{r^{2}}-1)d\nu^{2}+2drd\nu,\label{eq:2D p-RN in EF}
\end{equation}

~\\ with $\nu$ being periodic with period $\alpha$, i.e., the metric
is defined on $M:=S^{1}\times(0,\infty)$. Although we have removed
the coordinate singularities, the light cones still tilt over at $r=r_{\pm}$,
as we will discuss later. The horizon function $H(r)$ given by

\begin{gather*}
H(r)=\frac{2m}{r}-\frac{q^{2}}{r^{2}}-1
\end{gather*}

~\\ determines the horizons at $r=r_{\pm}$. In contrast to the pseudo-Schwarzschild
spacetime, there are three distinct cases:

~

(i) For $m>q$, there are two horizons: The event horizon, which is
also a Cauchy horizon, at
\begin{gather*}
r=r_{-}=m-\sqrt{m^{2}-q^{2}},
\end{gather*}
and a Cauchy horizon at

\begin{gather*}
r=r_{+}=m+\sqrt{m^{2}-q^{2}}.
\end{gather*}
Above all, both horizons are chronology horizons.

~

(ii) For $m=q$, there is single horizon at
\begin{gather*}
r=r_{+}=r_{-}=m=q,
\end{gather*}
which serves as both a Cauchy horizon and an event horizon.

~

(iii) For $m<q$, no horizon exists, which in the Reissner-Nordstr\"{o}m
spacetime results in a \textit{naked singularity} observable from
the outside. Although an event horizon is absent, the future singularity
at $r=0$ remains spacelike and thus remains hidden from external
observers. This scenario parallels the pseudo-Schwarzschild solution
with negative mass, a topic we will not explore further here.

~

The Riemann curvature tensor for the $2$-dimensional pseudo-Reissner-Nordstr\"{o}m
spacetime can be obtained using the same approach as in Section~\ref{sec:Pseudo-Schwarzschild-Spacetime}
for the pseudo-Schwarzschild case. The scalar curvature then follows
immediately 

\begin{gather*}
S=g^{jl}Ric_{jl}=\sum_{jl\,}g^{jl}Ric_{jl}=-2\cdot\frac{2mr-3q^{2}}{r^{4}}.
\end{gather*}

~\\ Thus, the Gaussian curvature $K$ is given by
\begin{gather*}
K=\frac{1}{2}S=-\frac{2mr-3q^{2}}{r^{4}}=-R_{1212}.
\end{gather*}
~\\ This result confirms that the Gaussian curvature changes sign
based on the value of $r$: for $r<\frac{3q^{2}}{2m}$, we find $2mr-3q^{2}<0$,
making the Gaussian curvature $K>0$. Conversely, for $r>\frac{3q^{2}}{2m}$,
we have $2mr-3q^{2}>0$, resulting in $K<0$.

~

To obtain the curvature for the $4$-dimensional case, we can derive
the Riemann curvature tensor for the pseudo-Reissner-Nordstr\"{o}m
spacetime by applying a Wick rotation to the well-known Reissner-Nordstr\"{o}m
spacetime. This approach uses the established properties of the Reissner-Nordstr\"{o}m
solution, allowing us to modify the spacetime signature and adapt
the curvature properties for the pseudo-Reissner-Nordstr\"{o}m framework.

\subsection{Pseudo-Reissner-Nordstr\"{o}m spacetime with two horizons\label{subsec:Pseudo-Reissner-Nordstrom- 2 horizons}}

The pseudo-Reissner-Nordstr\"{o}m spacetime $M=S^{1}\times(0,\infty)$
with $m>q$ and two horizons at $r=r_{\pm}$ presents a particularly
interesting configuration. These horizons divide $M$ into three distinct
regions: $M_{3}=\{(\nu,r)\in M\mid0<r<r_{-}\}$, $M_{2}=\{(\nu,r)\in M\mid r_{-}<r<r_{+}\}$,
and $M_{1}=\{(\nu,r)\in M\mid r_{+}<r<\infty\}$. In analogy with
the pseudo-Schwarzschild spacetime, the coordinate vector field $\partial_{r}$
is null, with $-\partial_{r}$ designated as future-pointing (indicating
that $r$ decreases as one moves from past to future). The behavior
of the vector field $\partial_{\nu}$ varies across regions: it is
spacelike in $M_{1}$, timelike in $M_{2}$, and again spacelike in
$M_{3}$.

~\\ For radial null geodesics, we first set

\begin{gather*}
ds^{2}=-(\frac{2m}{r}-\frac{q^{2}}{r^{2}}-1)\,d\nu^{2}+2drd\nu=0,
\end{gather*}

~\\ from which we immediately see 

\begin{gather*}
\frac{d\nu}{dr}=\begin{cases}
\begin{array}{c}
2(\frac{2m}{r}-\frac{q^{2}}{r^{2}}-1)^{-1}\\
0
\end{array} & \begin{array}{c}
(outgoing)\:\\
(infalling)
\end{array}\end{cases}.
\end{gather*}

~\\ In the next step we consider the outgoing radial null curves

\begin{gather*}
d\nu=\frac{2}{-1+\frac{2m}{r}-\frac{q^{2}}{r^{2}}}dr=\frac{2r^{2}}{R}dr=\frac{2r^{2}}{cr^{2}+br+a}dr,
\end{gather*}

~\\ where $R=2mr-q^{2}-r^{2}$, $a=-q^{2}$, $b=2m$ and $c=-1$.
This then takes the form of the integral $\nu(r)=2\int\frac{r^{2}}{R}dr=2(-r-m\ln R+(2m^{2}-q^{2})\int\frac{1}{R}dr)$,
where $\triangle=4ac-b^{2}=4(q^{2}-m^{2})<0$, and

\begin{gather*}
\int\frac{1}{R}dr=\frac{1}{\sqrt{-\triangle}}\ln\frac{\sqrt{-\triangle}-(b+2cr)}{(b+2cr)+\sqrt{-\triangle}}=\frac{1}{\sqrt{-4(q^{2}-m^{2})}}\ln\frac{\sqrt{-4(q^{2}-m^{2})}-(2m-2r)}{(2m-2r)+\sqrt{-4(q^{2}-m^{2})}}.
\end{gather*}

~\\ From this we obtain $\nu=\textrm{const}$ and 

\begin{gather*}
\nu(r)=2\left(-r-m\ln(2mr-q^{2}-r^{2})+(2m^{2}-q^{2})(\frac{1}{\sqrt{-4(q^{2}-m^{2})}}\ln\frac{\sqrt{-4(q^{2}-m^{2})}-(2m-2r)}{(2m-2r)+\sqrt{-4(q^{2}-m^{2})}})\right).
\end{gather*}

Consequently, for large $r$, the slope of the light cones is $\frac{d\nu}{dr}=\begin{cases}
\begin{array}{c}
-2\\
0
\end{array} & ,\end{cases}$ while as we approach $r_{+}$, the slope becomes $\frac{d\nu}{dr}=\begin{cases}
\begin{array}{c}
\rightarrow-\infty\\
0
\end{array} & ,\end{cases}$ and the light cones begin to open up once we cross the horizon (as
illustrated in Figure~\ref{fig:pseudo-Reissner-Nordstrom 2horizons}).
As we move closer to $r_{-}$, the light cones close up, giving $\frac{d\nu}{dr}=\begin{cases}
\begin{array}{c}
\rightarrow+\infty\\
0
\end{array} & .\end{cases}$ We can traverse this horizon, and as the light cones tilt over, approaching
the singularity at $r=0$, the slope becomes $\frac{d\nu}{dr}=\begin{cases}
\begin{array}{c}
0\\
0
\end{array} & .\end{cases}$

~

In regions $M_{3}$ and $M_{1}$, future-pointing paths move in the
direction of decreasing $r$, while in region $M_{2}$ this behavior
conflicts with causality. The surfaces at $r=r_{\pm}$ act as both
chronology horizons and event horizons; crossing these horizons signifies
a point of no return. For a particle initially starting its journey
in region $M_{1}$, encountering the horizon at $r=r_{+}$ is inevitable,
while it is possible to accelerate away from the horizon at $r=r_{-}$.

\begin{figure}[H]
\centering{}\includegraphics[scale=0.43]{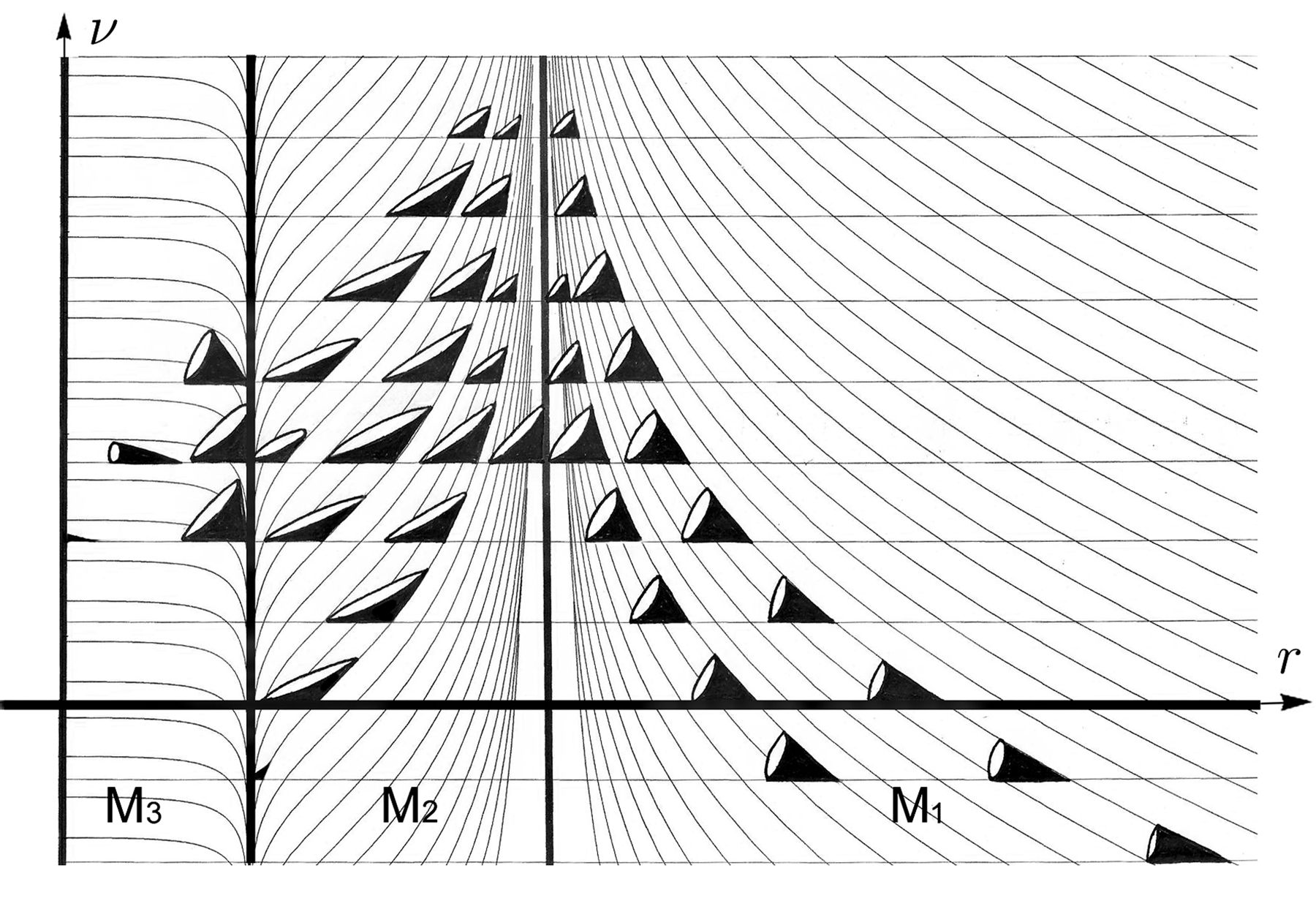}\caption{{\small{}\label{fig:pseudo-Reissner-Nordstrom 2horizons}Spacetime
diagram for the pseudo-Reissner-Nordstr\"{o}m spacetime with two
horizons.}}
\end{figure}

\subsubsection{\textup{Non-chronal region $M_{2}$}\textbf{ }}

Although the closed timelike curves (CTCs) in this case primarily
arise from the periodicity of the $\nu$-coordinate, the tilting of
light cones also contributes to the formation of CTCs.
\begin{claim}
A causal curve $\gamma:\,I\rightarrow M$ with velocity $\gamma'=\gamma_{1}'\partial_{\nu}+\gamma_{2}'\partial_{r}$
is future-pointing if and only if $\gamma_{1}'\geq0$.
\end{claim}

\begin{proof}
The inner product of two future-pointing causal vectors is non-positive.
We already defined $-\partial_{r}$ to be future-pointing, thus the
claim follows directly from $0\geq g(\gamma',-\partial_{r})=-\gamma_{1}'\Longleftrightarrow\gamma_{1}'\geq0$.
\end{proof}
\begin{claim}
\label{lem:p-RN - curve causal + future pointing}A curve $\gamma:I\rightarrow M$
with velocity $\gamma'=\gamma_{1}'\partial_{\nu}+\gamma_{2}'\partial_{r}$
is causal and future-pointing in $M_{2}$ if and only if $\frac{2\cdot\gamma_{2}'}{(\frac{2m}{r}-\frac{q^{2}}{r^{2}}-1)}\leq\gamma_{1}'$.
In $M_{1}$ and $M_{3}$, it is causal and future-pointing if and
only if $\frac{2\cdot\gamma_{2}'}{(\frac{2m}{r}-\frac{q^{2}}{r^{2}}-1)}\geq\gamma_{1}'\geq0$.
\end{claim}

\begin{proof}
A future-pointing causal curve is a curve $\gamma:\,I\rightarrow M$
satisfying that $g(\gamma',\gamma')\leq0$ and $\gamma_{1}'\geq0$.
Hence,

~\\ $0\geq g(\gamma',\gamma')=g(\gamma_{1}'\partial_{\nu}+\gamma_{2}'\partial_{r},\,\gamma_{1}'\partial_{\nu}+\gamma_{2}'\partial_{r})=\gamma_{1}'[2\cdot\gamma_{2}'-\gamma_{1}'\cdot(\frac{2m}{r}-\frac{q^{2}}{r^{2}}-1)]$\\$\Longleftrightarrow2\cdot\gamma_{2}'-\gamma_{1}'\cdot(\frac{2m}{r}-\frac{q^{2}}{r^{2}}-1)\leq0\Longleftrightarrow2\cdot\gamma_{2}'\leq\gamma_{1}'\cdot(\frac{2m}{r}-\frac{q^{2}}{r^{2}}-1)$.

~\\ Now we have to discriminate between two cases:

~

\begin{tabular}{ll}
i) & $r_{-}<r<r_{+}\Longrightarrow\,(\frac{2m}{r}-\frac{q^{2}}{r^{2}}-1)>0$:
$\frac{2\cdot\gamma_{2}'}{(\frac{2m}{r}-\frac{q^{2}}{r^{2}}-1)}\leq\gamma_{1}'$\tabularnewline
 & \tabularnewline
ii) & $0<r<r_{-}\,\vee\,r_{+}<r\Longrightarrow$ $(\frac{2m}{r}-\frac{q^{2}}{r^{2}}-1)<0$:
$\frac{2\cdot\gamma_{2}'}{(\frac{2m}{r}-\frac{q^{2}}{r^{2}}-1)}\geq\gamma_{1}'$,\tabularnewline
\end{tabular}

~\\ and this is precisely the assertion of the claim.
\end{proof}
The primary significance of this result lies in its identification
of regions where closed timelike curves (CTCs) may occur, thereby
justifying a focus on region $M_{2}$. The earlier identification
of $\nu=0$ and $\nu=\alpha$ leads directly to this outcome. On the
hypersurfaces $H_{r}$, where $r=\textrm{const}$, the metric reduces
to 

\begin{gather*}
g_{H_{r}}=-(\frac{2m}{r}-\frac{q^{2}}{r^{2}}-1)d\nu^{2},
\end{gather*}

~\\ which means that these surfaces are unconditionally spacelike
when $0<r<r_{-}\,\vee\,r_{+}<r$. For a fixed value of $r$, the circle
of constant radius $\gamma_{r}=\{(\nu,r)\mid0\leq\nu\leq\alpha\}$
is timelike if $r_{-}<r<r_{+}$, spacelike if $0<r<r_{-}\,\vee\,r_{+}<r$,
and null for $r=r_{\pm}$. This implies that a curve $\gamma(s)=(s,r)$
with constant $r\in(r_{-},r_{+})$ forms a closed timelike curve due
to the periodicity of $\nu$, and when $r=r_{\pm}$, it forms a closed
lightlike curve. Thus, we refer to $r=r_{\pm}$ as the chronology
horizon, and we designate the surfaces $H_{-}\coloneqq\{(\nu,r_{-})\mid0\leq\nu\leq P\}$
and $H_{+}:=\{(\nu,r_{+})\mid0\leq\nu\leq P\}$, with constant radii
$r_{-}$ and $r_{+}$, respectively, as the null hypersurfaces.
\begin{prop}
A closed timelike curve in $M$ must lie entirely within region $M_{2}$.
\end{prop}

\begin{proof}
From what has already been proved, for a timelike future-pointing
curve in $M_{1}$ we have the requirement $\frac{2}{(\frac{2m}{r}-\frac{q^{2}}{r^{2}}-1)}<\frac{\gamma_{1}'}{\gamma_{2}'}$.
The procedure is to exclude that a timelike curve starts in region
$M_{1}$, enters the region $M_{2}$ and loops back to its initial
point in $M_{1}$. Let us consider a timelike, future-pointing curve
$\gamma:[0,1]\rightarrow M$ starting in $M_{1}$. The tangent vector
to the curve is given by $\gamma'=\gamma_{1}'\partial_{\nu}+\gamma_{2}'\partial_{r}$,
where $\gamma_{1}'$ and $\gamma_{2}'$ are the components of the
tangent vector. Since the fundamental vector field $\partial_{r}$
is null everywhere, and given that $\gamma$ is future-pointing and
timelike, we have $\gamma_{1}'>0$ along the entire curve $\gamma$,
and $\gamma_{2}'<0$ for $\gamma\subset M_{1}$.

~\\ Assume $\gamma(0)\in M_{1}$, then $\underset{<0}{\underbrace{\frac{2}{(\frac{2m}{r}-\frac{q^{2}}{r^{2}}-1)}}}<\frac{\gamma_{1}'}{\gamma_{2}'}$
follows from the aforegoing Lemma. As we approach the Cauchy horizon
$H_{+}$ at $r_{+}$, we have $(\frac{2m}{r}-\frac{q^{2}}{r^{2}}-1)\longrightarrow0$,
and thus $\frac{2}{(\frac{2m}{r}-\frac{q^{2}}{r^{2}}-1)}\longrightarrow-\infty$.
Consequently, $-\infty<\frac{\gamma_{1}'}{\gamma_{2}'}$ as we approach
$H_{+}$. The curve would remain in $M_{1}$ and not enter region
$M_{2}$ only if $\gamma_{2}'=0$, meaning that it would be tangent
to $\partial_{\nu}$ and thus lightlike. However, since $\gamma$
is timelike and $\gamma_{2}'\neq0$, it follows that the curve must
enter the region $M_{2}$. 

~\\ In region $M_{2}$, we have the requirement $\gamma_{2}'<\frac{\gamma_{1}\cdot(\frac{2m}{r}-\frac{q^{2}}{r^{2}}-1)}{2}$.
As $r\longrightarrow r_{+}$, we observe that $\frac{\gamma_{1}\cdot(\frac{2m}{r}-\frac{q^{2}}{r^{2}}-1)}{2}\longrightarrow0$,
and thus we can assert that $\gamma_{2}'<0$ close to the horizon
$H_{+}$. This shows that $\gamma$ cannot cross the horizon $H_{+}$
anymore and is therefore confined to region $M_{2}$. Consequently,
$\gamma$ cannot return to its initial position. In particular, it
is evident that there are no closed timelike curves (CTCs) in region
$M_{1}$, and similar considerations apply to region $M_{3}$. 

~\\ It is immediately clear that a curve $\gamma(s)=(s,r)$ with
constant $r\in(r_{-},r_{+})$ forms a closed timelike curve. According
to Proposition~\ref{pro: infinite CTCs}, there are infinitely many
CTCs in region $M_{2}$.
\end{proof}
The above arguments establish the following 
\begin{cor}
There are closed timelike curves in $M_{2}$, and for every point
$p\in M_{2}$, there exists a closed timelike curve passing through
$p$.
\end{cor}

\begin{proof}
Under the conditions stated above, consider $\gamma:[0,L]\rightarrow M,\gamma(s)=(ks,r),\,k\geq0.$
\end{proof}
\begin{figure}[H]
\centering{}\includegraphics[scale=0.33]{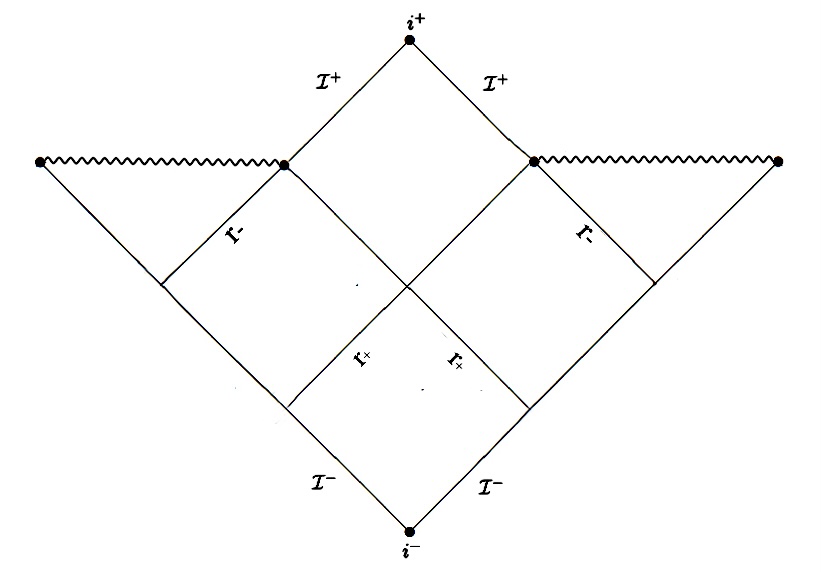}\caption{{\small{}\label{fig:Penrose pseudo-Reissner-Nordstrom}Conformal diagram
for the pseudo-Reissner-Nordstr\"{o}m spacetime with two horizons.}}
\end{figure}

\subsection{Extremal pseudo-Reissner-Nordstr\"{o}m spacetime\textmd{ }}

The extremal $2$-dimensional pseudo-Reissner-Nordstr\"{o}m spacetime
$M$ with $m=q$ has only one horizon at $r=m$, which separates $M$
into two regions: $M_{3}=\{(\nu,r)\in M\mid0<r<m\}$ and $M_{2}=\{(\nu,r)\in M\mid m<r\}$.
As before, we time-orient $M$ by requiring $-\partial_{r}$ to be
future pointing. The metric degenerates to 

\begin{gather*}
ds^{2}=-(\frac{2m}{r}-\frac{m^{2}}{r^{2}}-1)d\nu^{2}+2drd\nu=(\frac{m^{2}-2mr+r^{2}}{r^{2}})d\nu^{2}+2drd\nu=\frac{(m-r)^{2}}{r^{2}}d\nu^{2}+2drd\nu.
\end{gather*}

The Gaussian curvature is

\begin{gather*}
K=\frac{1}{2}S=-\frac{m(2r-3m)}{r^{4}}=-R_{trtr},
\end{gather*}

~\\ where for $r<\frac{3m}{2}$, we have $2mr-3m^{2}<0$, so the
Gaussian curvature is $\frac{2mr-3m^{2}}{r^{4}}>0$, and for $r>\frac{3m}{2}$,
the curvature becomes negative.

~\\ By the same method as in Subsection~\ref{subsec:Pseudo-Reissner-Nordstrom- 2 horizons},
we calculate the radial null geodesics by setting

\begin{gather*}
ds^{2}=-(\frac{2m}{r}-\frac{m^{2}}{r^{2}}-1)d\nu^{2}+2drd\nu=0,
\end{gather*}
and we obtain

\begin{gather*}
\frac{d\nu}{dr}=\begin{cases}
\begin{array}{c}
2(\frac{2m}{r}-\frac{m^{2}}{r^{2}}-1)^{-1}\\
0
\end{array} & \begin{array}{c}
(outgoing)\;\\
(infalling)
\end{array}\end{cases}.
\end{gather*}
~\\ This term can be handled in much the same way as in Section~\ref{subsec:Pseudo-Reissner-Nordstrom- 2 horizons },
the only difference being that $q=m$. A quick verification shows
that we obtain the integral $\nu=\textrm{const}$, and the expression
for $\nu(r)$ is 

\begin{gather*}
\nu(r)=2\left(-r+2m\ln(r-m)-\frac{m^{2}}{m-r}\right).
\end{gather*}

~\\ Since $g(\partial_{r},\partial_{r})=0$ for all $r\in(0,\infty)$,
we conclude that $\partial_{r}$ is lightlike. Additionally, from
$g(\partial_{u},\partial_{u})=\frac{(m-r)^{2}}{r^{2}}>0$ for all
$r\in(0,\infty)$, we can deduce that $\partial_{u}$ is spacelike
at any point $p\in M$. See Figure~\ref{fig: pseudo-Reissner-Nordstrom 1 horizon Light-cone-structure}.

\begin{figure}[H]
\centering{}\includegraphics[scale=0.43]{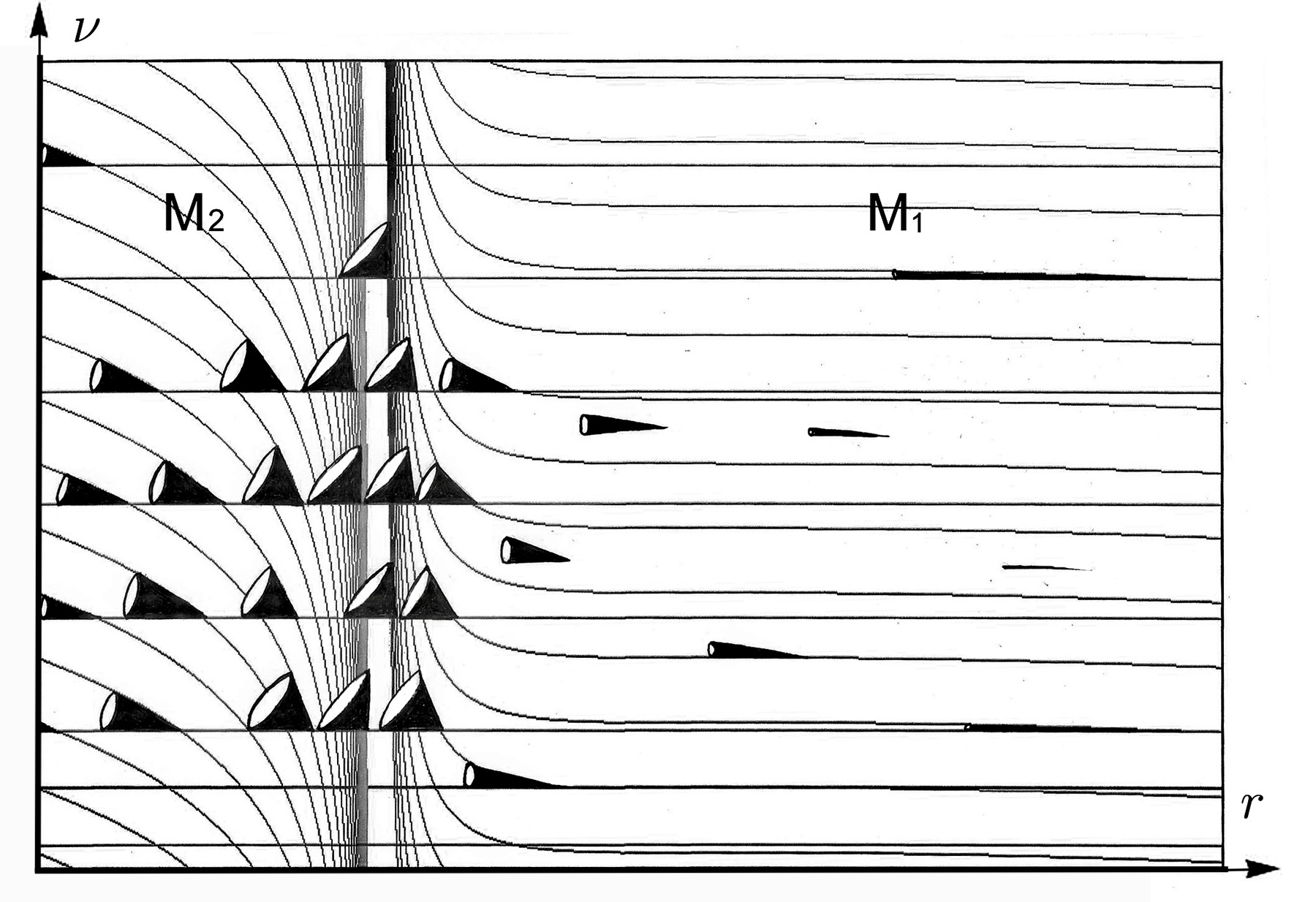}\caption{{\small{}\label{fig: pseudo-Reissner-Nordstrom 1 horizon Light-cone-structure}Light
cone structure for the extremal pseudo-Reissner-Nordstr\"{o}m spacetime.}}
\end{figure}

\begin{prop}
A curve $\gamma:\,I\rightarrow M$ with velocity \textup{$\gamma'=\gamma_{1}'\partial_{\nu}+\gamma_{2}'\partial_{r}$}
is causal and future-pointing in $M$ if and only if $2\cdot\gamma_{2}'\leq-\gamma_{1}'\cdot\frac{(m-r)^{2}}{r^{2}}$.
\end{prop}

\begin{proof}
A future-pointing causal curve is a curve $\gamma:\,I\rightarrow M$,
satisfying that $g(\gamma',\gamma')\leq0$ and $\gamma_{1}'\geq0$.
Hence,

~

\begin{tabular}{ll}
 & $0\geq g(\gamma',\gamma')=g(\gamma_{1}'\partial_{\nu}+\gamma_{2}'\partial_{r},\,\gamma_{1}'\partial_{\nu}+\gamma_{2}'\partial_{r})=\gamma_{1}'[2\cdot\gamma_{2}'+\gamma_{1}'\cdot\frac{(m-r)^{2}}{r^{2}}]$\tabularnewline
 & \tabularnewline
 & $\Longleftrightarrow2\cdot\gamma_{2}'+\gamma_{1}'\cdot\frac{(m-r)^{2}}{r^{2}}\leq0\Longleftrightarrow2\cdot\gamma_{2}'\leq\underset{\leq0}{\underbrace{-\gamma_{1}'\cdot\frac{(m-r)^{2}}{r^{2}}}}$.\tabularnewline
\end{tabular}

~
\end{proof}
We are now in the position to show 
\begin{prop}
All future pointing causal curves in $M_{2}$ converge towards the
singularity $r=0$.
\end{prop}

\begin{proof}
For a causal future-pointing curve $\gamma:[0,1]\rightarrow M_{2}:\,s\longmapsto\gamma(s)$
we have the requirement $\gamma_{1}\geq0$, and $\gamma_{2}'\leq\underset{\leq0}{\underbrace{-\gamma_{1}'\cdot\frac{(m-r)^{2}}{r^{2}}\cdot\frac{1}{2}}}$
$\Longrightarrow\gamma_{2}'\leq0$. For $\gamma_{2}'=0$ we have $0\leq\underset{\leq0}{\underbrace{-\gamma_{1}'\cdot\frac{(m-r)^{2}}{r^{2}}\cdot\frac{1}{2}}}\Longleftrightarrow\gamma_{1}'=0\,\vee\,\frac{(m-r)^{2}}{r^{2}}=0$.
Note that $\gamma_{1}'=0$ implies, on the one hand, that $\gamma'=\gamma_{1}'\partial_{\nu}+\gamma_{2}'\partial_{r}=0$,
which means $\gamma'$ is not future pointing. On the other hand,
the solution to $\frac{(m-r)^{2}}{r^{2}}=0$ gives $r=m$, which is
the horizon and is not part of region $M_{2}$. Hence, we conclude
that $\frac{dr}{ds}=\gamma_{2}'<0$, as claimed.
\end{proof}

\begin{figure}[H]
\centering{}\includegraphics[scale=0.33]{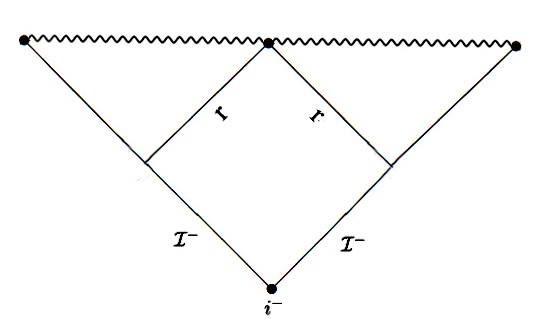}\caption{{\small{}Conformal diagram for the extremal pseudo-Reissner-Nordstr\"{o}m
spacetime.}}
\end{figure}

From these results, and given that $r$ is strictly monotonic decreasing,
we can conclude that for a future-pointing timelike curve---as defined
above---it is impossible to form a closed curve. Therefore, the extremal
pseudo-Reissner-Nordstr\"{o}m spacetime does not contain closed timelike
curves (CTCs). A closed lightlike curve exists only at $r=m$. Although
the extremal pseudo-Reissner-Nordstr\"{o}m spacetime does not permit
closed timelike curves, a CTC almost emerges at $r=m$. If the light
cones were to open up slightly more at that point, chronology would
be violated. 

\subsection{Pseudo-Reissner-Nordstr\"{o}m timelike geodesics}

We first observe that the geodesics exhibit behavior similar to those
in the pseudo-Schwarzschild or Misner space. This becomes evident
when we write down the explicit timelike geodesic equation as a function
of the radius by substituting $g_{\nu\nu=}-(\frac{2m}{r}-\frac{q^{2}}{r^{2}}-1)$
into Equation~(\ref{eq:timelike geodesic}) in Section~\ref{sec:General-Equation-forGeodesics}:

\begin{equation}
\nu(r)=\int\frac{\frac{-1}{(\frac{2m}{r}-\frac{q^{2}}{r^{2}}-1)}\left(\xi-\omega\sqrt{\xi^{2}-(\frac{2m}{r}-\frac{q^{2}}{r^{2}}-1)}\right)}{\omega\sqrt{\xi^{2}-(\frac{2m}{r}-\frac{q^{2}}{r^{2}}-1)}}dr.\label{eq:p-RN timelike geodesic}
\end{equation}

~\\ Such a timelike geodesic with $\xi=0$ encircles the horizon,
while a geodesic with $\xi>0$ crosses the horizon. For the latter
case we consider $\xi^{2}=\frac{2m}{r}-\frac{q^{2}}{r^{2}}-1\Longleftrightarrow(\xi^{2}+1)r^{2}-(2m)r+q^{2}=0$.

~\\ By virtue of the discriminant 

\begin{gather*}
D=4m^{2}-4q^{2}(\xi^{2}+1)\geq0\Longleftrightarrow m^{2}\geq q^{2}(\xi^{2}+1),
\end{gather*}

~\\ we obtain two potential turning points at

\begin{equation}
r_{turn1,2}=\frac{2m\pm\sqrt{4(m^{2}-q^{2}(\xi^{2}+1))}}{2(\xi^{2}+1)}=\frac{m\pm\sqrt{m^{2}-q^{2}(\xi^{2}+1)}}{(\xi^{2}+1)}.
\end{equation}

~

As our focus in this work is to elaborate on and demonstrate the similarities
between the three different spacetimes, we will disregard $r_{turn2}$
and instead direct our attention to the region $M_{1}\cup M_{2}$,
specifically to the turning point contained within that region:

\begin{equation}
r_{turn1}=\frac{m+\sqrt{m^{2}-q^{2}(\xi^{2}+1)}}{(\xi^{2}+1)}.
\end{equation}

~

We observe an analogous situation to the pseudo-Schwarzschild case,
where we obtain the corresponding turning point $r_{turn}=\frac{2m}{(\xi^{2}+1)}$
when setting $q=0$. The behavior of the null geodesics can be derived
from the equation 

\begin{gather*}
\nu(r)=\int\frac{-1}{(\frac{2m}{r}-\frac{q^{2}}{r^{2}}-1)}\left(\omega-1\right)dr.
\end{gather*}

~

Analogous to the previously discussed cases, we have a set of null
geodesics, specifically $\nu(r)=\textrm{const}$, that crosses the
surfaces at $r=r_{\pm}$. These geodesics are complete and extend
from $r=0$ to $r=\infty$. In contrast, another set of geodesics
spirals around the horizons at $r=r_{\pm}$, and as a result, these
geodesics are incomplete within this coordinate patch.

~

Due to the geodesic incompleteness, the embedded pseudo-Reissner-Nordstr\"{o}m
cylinder with $\theta=\textrm{const}$ and $\phi=\textrm{const}$,
restricted to $M_{1}\cup M_{2}$, resembles the situation in the $2$-dimensional
Misner space and cylindrical pseudo-Schwarzschild spacetime.

\section{Isocausality}\label{sec:Isocausality}

In this section we investigate the causal equivalence properties---in the sense of 
\emph{isocausality} as introduced by Garc\'{i}a-Parrado and Senovilla---between the three 
model spacetimes considered in this work: the Misner spacetime $(\mathcal M_{\mathrm{M}},g_{\mathrm{M}})$, 
the pseudo-Schwarzschild spacetime $(\mathcal M_{\mathrm{PS}},g_{\mathrm{PS}})$, 
and the pseudo-Reissner-Nordstr\"{o}m spacetime $(\mathcal M_{\mathrm{PRN}},g_{\mathrm{PRN}})$. 

\begin{definition}[Isocausality]
Two time-oriented spacetimes $(\mathcal M,g)$ and $(\mathcal N,h)$ are 
\emph{isocausal} if there exist smooth maps
\[
\Phi:\mathcal M \longrightarrow \mathcal N,
\qquad
\Psi:\mathcal N \longrightarrow \mathcal M
\]
such that both $\Phi$ and $\Psi$ are \emph{causal}, i.e.\ their differentials map 
future-directed causal vectors to future-directed causal vectors. 
No bijectivity or inverse relationship between $\Phi$ and $\Psi$ is required by this definition.
\end{definition}

In several results below we prove a \emph{stronger} property: 
the causal maps $\Phi$ and $\Psi$ are actually mutual diffeomorphisms 
on the universal covers and on certain causally regular regions of the spacetimes. 
For the compactified models, global mutual diffeomorphisms exist only under an explicit 
equivariance condition relating the deck group actions; without this condition, 
isocausality in the strict sense is generally not guaranteed.

\medskip
Each of our model spacetimes is a warped product 
\[
\mathcal M_i \;=\; \mathcal M^{(Z)}_i \times_{r} H^2
\]
with $2$-dimensional Lorentzian base $(\mathcal M^{(Z)}_i,g^{(Z)}_i)$ 
admitting coordinates $(\nu,r)$ in which the metric takes the canonical 
Eddington-Finkelstein-type form
\begin{equation}
    g^{(Z)}_i \;=\; f_i(r)\,d\nu^2 + 2\,d\nu\,dr,
    \label{eq:IsoIntroCanonical}
\end{equation}
and $H^2$ the hyperbolic space of constant negative curvature endowed with 
its standard metric $g_{H^2}$.
The explicit functions $f_i(r)$ for $i=\mathrm{M},\mathrm{PS},\mathrm{PRN}$ 
have been given in the preceding sections.

\medskip
Because every $1+1$ Lorentzian metric is locally conformally flat, 
any two bases $g^{(Z)}_i$ and $g^{(Z)}_j$ are locally related by a smooth, strictly positive conformal factor 
via a suitable change of null coordinates. 
By combining this with a careful control of the warp factor $r$, 
we obtain $4$-dimensional causal maps that preserve the full light-cone structure (no restriction to base-only vectors). 
This observation motivates the main results of this section:

\begin{itemize}
    \item In Proposition~\ref{prop:IsoAllThree} we construct explicit causal bijections between 
    the universal covers of the three models, proving that they are \emph{pairwise isocausal} in the covering space. 
    We also identify causally regular regions in the original spacetimes on which the same holds.
    \item Part~(C) of Proposition~\ref{prop:IsoAllThree} provides a clean necessary and sufficient condition 
    (equivariance with respect to the deck group actions and single-valuedness of the conformal factor) 
    under which the isocausality descends to the compactified spacetimes.
    \item Corollary~\ref{cor:NoEquivariance} clarifies that, in the absence of such equivariance, 
    one obtains only a one-way causal relation between the compactified models.
    \item Remark~\ref{rem:DefaultOneWay} emphasizes that this one-way causal relation 
    should be regarded as the \emph{generic} global situation unless the deck-group matching is verified explicitly.
\end{itemize}

\begin{prop}[Isocausality of the Misner-type family]\label{prop:IsoAllThree}
Let $(\mathcal M_{\mathrm{M}},g_{\mathrm{M}})$ be Misner spacetime, $(\mathcal M_{\mathrm{PS}},g_{\mathrm{PS}})$ the
pseudo-Schwarzschild spacetime, and $(\mathcal M_{\mathrm{PRN}},g_{\mathrm{PRN}})$ the
pseudo-Reissner-Nordstr\"om spacetime, with the $1+1$ radial (base) metrics written in the canonical
Eddington-Finkelstein-type form \eqref{eq:IsoIntroCanonical}, and the full $4$-metrics given by
\[
g_i \;=\; g^{(Z)}_i \,\oplus\, r^2 g_{H^2}.
\]
Denote universal covers by $(\widetilde{\mathcal M}_i,\tilde g_i)$ and the covering projections by $\pi_i$.

\medskip
\noindent\textbf{(A) Local / covering isocausality.} 
For every pair $(i,j)\in\{\mathrm{M},\mathrm{PS},\mathrm{PRN}\}^2$
there exists a smooth bijection
\[
\widetilde\Phi_{ij}:\widetilde{\mathcal M}_i\longrightarrow\widetilde{\mathcal M}_j
\]
such that both $\widetilde\Phi_{ij}$ and $\widetilde\Phi_{ij}^{-1}$ map future-directed causal vectors to 
future-directed causal vectors. In particular the universal covers $(\widetilde{\mathcal M}_i,\tilde g_i)$ are pairwise isocausal.

\medskip
\noindent\textbf{(B) Isocausality on causally regular regions.} 
Let $U_i\subset\mathcal M_i$ denote open regions obtained by
removing (a) neighborhoods of curvature singularities, and (b) neighborhoods of null hypersurfaces where $f_i=0$,
so that $f_i$ is smooth and nowhere zero on $U_i$ and the $S^1$-orbits (if present) can be unwrapped. Then for every
pair $(i,j)$ there exists a diffeomorphism $\Phi_{ij}:U_i\to U_j$ such that both $\Phi_{ij}$ and $\Phi_{ij}^{-1}$
preserve future-directed causal vectors; i.e.\ $(U_i,g_i)$ and $(U_j,g_j)$ are isocausal.

\medskip
\noindent\textbf{(C) Global isocausality (quotient) criterion.} 
Suppose the covering-space diffeomorphism $\widetilde\Phi_{ij}$ 
constructed in (A) is \emph{equivariant} with respect to the deck groups (periodic identifications) 
\(\Gamma_i,\Gamma_j\) in the sense that there exists some integer $k\neq 0$ such that
\[
\widetilde\Phi_{ij}\circ\gamma \;=\; \gamma^k\circ\widetilde\Phi_{ij}
\quad\text{for all }\gamma\in\Gamma_i.
\]
If moreover the induced conformal factor descends to a smooth, strictly positive single-valued function on the quotient,
then $\widetilde\Phi_{ij}$ descends to a smooth causal map 
\[
\Phi_{ij}:\mathcal M_i\longrightarrow\mathcal M_j.
\]
If $|k|=1$, this descended map is a global causal bijection with causal inverse, so that $(\mathcal M_i,g_i)$ and $(\mathcal M_j,g_j)$ 
are globally isocausal. If $|k|>1$ the descended map is a covering of degree $|k|$ and hence defines only a one-way causal relation 
$\mathcal M_i\prec\mathcal M_j$.
\end{prop}

\begin{proof}
We split the argument into three steps. Throughout, we work on open base domains where $f_i$ and $f_j$
are smooth and nonvanishing; this excludes only neighborhoods of horizons ($f=0$) and singularities and suffices for local and covering statements.

\smallskip
\emph{Step 1 (base-level conformal bijection).}
Write each base metric in \emph{ingoing EF form} with null coordinates $(u_i,v_i)=(\nu_i,r_i)$:
\[
g^{(Z)}_i \;=\; f_i(v_i)\,du_i^2 + 2\,du_i\,dv_i,
\qquad v_i=r_i.
\]
Fix a constant $a>0$ and seek a base diffeomorphism of the form
\[
\phi_{ij}(u_i,v_i) \,=\, (u_j,v_j) \,=\, (a\,u_i,\ \psi(v_i)),
\]
where $\psi$ is a smooth strictly monotone solution of the first-order ODE
\begin{equation}\label{eq:ODEpsi}
\psi'(v) \,=\, a\,\frac{f_j(\psi(v))}{f_i(v)}.
\end{equation}
Since $f_i$ and $f_j$ are smooth and nowhere vanishing on the chosen domains, \eqref{eq:ODEpsi} admits unique local (and, on each connected component, global) diffeomorphic solutions $\psi$ for any initial condition in the target range; moreover $\psi$ is strictly monotone with the same sign as $a f_j/f_i$.

A direct pullback computation then gives
\[
\phi_{ij}^* g^{(Z)}_j
= f_j(\psi)\,a^2\,du_i^2 + 2\,a\,\psi'(v_i)\,du_i\,dv_i
= a\,\psi'(v_i)\,\big( f_i(v_i)\,du_i^2 + 2\,du_i\,dv_i\big)
= \Omega_{ij}\, g^{(Z)}_i,
\]
with strictly positive conformal factor
\begin{equation}\label{eq:OmegaZ}
\Omega_{ij}(u_i,v_i) \,=\, a\,\psi'(v_i) \,=\, a^2\,\frac{f_j(\psi(v_i))}{f_i(v_i)} \;>\; 0,
\end{equation}
where we used \eqref{eq:ODEpsi} in the second equality. Hence $\phi_{ij}$ is a causal conformal diffeomorphism between the bases, with inverse obtained by solving the analogous ODE interchanging $i$ and $j$.

\smallskip
\emph{Step 2 (extension to the warped product and cone inclusion).}
Extend $\phi_{ij}$ trivially on the hyperbolic fiber:
\[
\widetilde\Phi_{ij}(u_i,v_i,\xi)\;=\;(\,a u_i,\ \psi(v_i),\ \xi\,),\qquad \xi\in H^2.
\]
With $g_i=g^{(Z)}_i\oplus v_i^2 g_{H^2}$ and $g_j=g^{(Z)}_j\oplus v_j^2 g_{H^2}$ (recall $v=r$), we obtain
\begin{equation}\label{eq:pullback4D}
\widetilde\Phi_{ij}^* g_j
\;=\; \Omega_{ij}\, g^{(Z)}_i \;\oplus\; \psi(v_i)^2\; g_{H^2}.
\end{equation}
Let $X=(X_Z,X_H)$ be any tangent vector at a point, with $X_Z\in T\mathcal M^{(Z)}_i$ and $X_H\in TH^2$.
If $X$ is future-directed causal for $g_i$ then
\[
g^{(Z)}_i(X_Z,X_Z)\;+\; v_i^2\,g_{H^2}(X_H,X_H)\ \le 0.
\]
Using \eqref{eq:pullback4D} and $\Omega_{ij}>0$ we estimate
\[
\widetilde\Phi_{ij}^* g_j(X,X)
= \Omega_{ij}\,g^{(Z)}_i(X_Z,X_Z) + \psi(v_i)^2\,g_{H^2}(X_H,X_H)
\le \big(\psi(v_i)^2 - \Omega_{ij}\,v_i^2\big)\,g_{H^2}(X_H,X_H).
\]
Therefore a \emph{sufficient} pointwise condition ensuring $\widetilde\Phi_{ij}$ is causal on the full $4$-space is
\begin{equation}\label{eq:keyineq}
\psi(v)^2 \ \le\ \Omega_{ij}(u_i,v)\, v^2
\qquad\text{for all }(u_i,v)\ \text{in the domain.}
\end{equation}
By \eqref{eq:OmegaZ}, \eqref{eq:keyineq} is equivalent to
\begin{equation}\label{eq:keyineq2}
\psi(v)^2 \ \le\ a\,\psi'(v)\,v^2.
\end{equation}

\paragraph{Ensuring \eqref{eq:keyineq2} on causally regular regions.}
Fix a connected base domain $I=[v_-,v_+]$ on which $f_i$ and $f_j$ are smooth and nowhere zero and with $0<v_-<v_+<\infty$ (i.e.\ away from $r=0$ and horizons/singularities). Solve \eqref{eq:ODEpsi} with a chosen initial value $\psi(v_0)=w_0$ so that $\psi$ maps $I$ diffeomorphically onto some interval $J\Subset(0,\infty)$. On $I$ and $J$ we have
\[
0<m_i\le f_i\le M_i,\qquad 0<m_j\le f_j\le M_j
\]
for positive constants $m_i,M_i,m_j,M_j$.
From \eqref{eq:ODEpsi}, $\psi'(v)\ge a\,m_j/M_i$.
Hence, for all $v\in I$,
\[
a\,\psi'(v)\,v^2 \ \ge\ \frac{a^2 m_j}{M_i}\,v^2 \ \ge\ \frac{a^2 m_j}{M_i}\,v_-^2.
\]
Moreover $\psi$ grows at most linearly in $a$ on $I$:
\[
|\psi(v)| \ \le\ |\psi(v_0)| \,+\, \int_{v_-}^{v_+}\!\! \psi'(s)\,ds
\ \le\ |\psi(v_0)| \,+\, a\,\frac{M_j}{m_i}\,(v_+-v_-)
\ :=\ C_0 \,+\, a\,C_1.
\]
Therefore $\psi(v)^2 \le 2C_0^2 + 2a^2 C_1^2$ on $I$, and \eqref{eq:keyineq2} will hold provided
\[
2C_0^2 + 2a^2 C_1^2 \ \le\ \frac{a^2 m_j}{M_i}\,v_-^2.
\]
This is achieved for all sufficiently large $a$ (since the $a^2$-coefficients satisfy $2C_1^2<\frac{m_j}{M_i}v_-^2$ after possibly shrinking $I$; and any fixed constant term $2C_0^2$ is dominated for large $a$).
Consequently, on any causally regular base interval $I$ bounded away from $0$ and the zeros of $f_i$, we can choose $a>0$ (and solve \eqref{eq:ODEpsi}) so that \eqref{eq:keyineq2}---hence \eqref{eq:keyineq}---holds. With this choice, $\widetilde\Phi_{ij}$ maps future-directed causal vectors for $g_i$ to future-directed causal vectors for $g_j$ in the full $4$-dimensional warped product.

As the inverse construction (interchanging $i$ and $j$ and repeating the argument on the corresponding regular interval) yields the reverse inequality on the target region, we obtain mutual causality and hence a causal bijection between the corresponding simply connected covered regions. This proves (A) on universal covers (which are unions of such simply connected regular blocks) and (B) on causally regular regions $U_i$.

\smallskip
\emph{Step 3 (descent under deck-equivariance).}
Let $\Gamma_i,\Gamma_j$ be the deck groups of the coverings $\pi_i,\pi_j$.
If $\widetilde\Phi_{ij}$ satisfies
\[
\widetilde\Phi_{ij}\circ\gamma=\gamma^k\circ\widetilde\Phi_{ij}\qquad(\gamma\in\Gamma_i)
\]
for some integer $k\neq0$, then $\widetilde\Phi_{ij}$ descends to a smooth map $\Phi_{ij}:\mathcal M_i\to\mathcal M_j$ with $\pi_j\circ\widetilde\Phi_{ij}=\Phi_{ij}\circ\pi_i$. Because the causal-inequality check \eqref{eq:keyineq} is invariant under the deck actions and the conformal factor $\Omega_{ij}$ is single-valued on the quotient by hypothesis, the descended differential $d\Phi_{ij}$ preserves future-directed causal vectors. If $|k|=1$ the descended map is a global diffeomorphism with causal inverse; if $|k|>1$ it is a covering of degree $|k|$ and so gives only a one-way causal relation $\mathcal M_i\prec\mathcal M_j$. This proves (C).
\end{proof}

\begin{corollary}[Default global outcome without equivariance]\label{cor:NoEquivariance}
Let $(\mathcal M_i,g_i)$ and $(\mathcal M_j,g_j)$ be any two of the Misner, pseudo-Schwarzschild, 
or pseudo-Reissner-Nordstr\"{o}m spacetimes. 
Suppose there exists an isocausality diffeomorphism 
\[
\widetilde\Phi_{ij}:\widetilde{\mathcal M}_i\longrightarrow\widetilde{\mathcal M}_j
\]
between their universal covers, as in Proposition~\ref{prop:IsoAllThree}(A), 
but that the equivariance hypothesis in Proposition~\ref{prop:IsoAllThree}(C) fails 
(e.g.\ the deck-group periods do not match up to an integer winding, 
or the conformal factor is not single-valued on the quotient). 
Then no well-defined descended bijection $\Phi_{ij}:\mathcal M_i\to\mathcal M_j$ exists. 
At best, one may construct a one-way global causal map 
\[
\mathcal M_i \prec \mathcal M_j,
\]
possibly of degree $|k|>1$ if partial equivariance holds, 
but not necessarily an isocausality relation in both directions.
\end{corollary}

\begin{proof}
Let $\pi_i,\pi_j$ be the coverings with deck groups $\Gamma_i,\Gamma_j$.
A map $\widetilde\Phi_{ij}$ descends to a single-valued map on the quotient if and only if it is equivariant with respect to the deck actions.
If equivariance fails, there exist $\gamma\in\Gamma_i$ and $\tilde x\in\widetilde{\mathcal M}_i$ such that
$\pi_j(\widetilde\Phi_{ij}(\tilde x))\neq \pi_j(\widetilde\Phi_{ij}(\gamma\tilde x))$; hence the composition
$\pi_j\circ\widetilde\Phi_{ij}\circ\pi_i^{-1}$ is not single-valued on $\mathcal M_i$ and does not define a global map.
If partial equivariance holds with degree $|k|>1$, the descended map is a covering (not bijective), so only the one-way causal relation $\mathcal M_i\prec\mathcal M_j$ remains.
\end{proof}

\begin{rem}[On the generic global situation]\label{rem:DefaultOneWay}
In the absence of an explicit verification of the deck-equivariance condition
in Proposition~\ref{prop:IsoAllThree}(C),
the generic expectation for the compactified Misner-type spacetimes considered here 
is that mutual isocausality will \emph{not} hold globally,
even though their universal covers are mutually isocausal.
Without equivariance, any global causal relationship must be established 
by direct construction. In particular, if the equivariance degree $|k|>1$,
then one obtains only a one-way causal relation $\mathcal M_i\prec\mathcal M_j$,
reflecting the fact that causal curves in $\mathcal M_i$ can be lifted into $\mathcal M_j$,
but the converse need not hold. 
Thus, unless $|k|=1$, the compactified models differ genuinely in their global causal structure.
\end{rem}

\medskip
\noindent\textbf{Physical meaning of global vs.\ covering isocausality.}
The distinction between isocausality at the level of the universal covers and the compactified spacetimes 
has a clear operational interpretation. On the universal cover, the causal structure is the ``local'' one perceived 
by any observer restricted to a simply connected region: the shape of the light cones, the accessibility of events 
along causal curves, and the arrangement of horizons are indistinguishable across the three models. 
In this sense, freely falling observers with access only to local measurements would not be able to tell which spacetime 
they inhabit; the causal physics is the same.

~

However, when passing to the compactified spacetimes by quotienting along closed orbits, additional global 
identifications introduce genuinely new causal phenomena, most notably the appearance and arrangement of closed 
timelike curves (CTCs). Failure of the deck-equivariance condition means that these identifications differ between the 
models: a causal curve that is closed in one model may fail to be closed in another, or the chronology-violating set 
may have a different extent. Operationally, an observer traversing a large loop along a timelike orbit in one spacetime 
might return to their starting event, while in another spacetime with the same local light cone structure they would not 
---or would reappear at a different event. Thus, the one-way causal relation 
$\mathcal M_i \prec \mathcal M_j$ reflects the fact that every causal path available in $\mathcal M_i$ is realisable 
in $\mathcal M_j$, but not necessarily the other way around. This is a genuinely global effect that cannot be detected 
from purely local measurements.

\section{Conclusion}

We have investigated a family of causality-violating spacetimes---specifically, the pseudo-Schwarzschild, 
pseudo-Reissner-Nordstr\"{o}m, and Misner spacetimes---that share fundamental causal structures, 
including the presence of chronology horizons, Cauchy horizons, radial geodesics, and acausal regions. 
All these spacetimes can be expressed in a warped product form, with metrics decomposing into the cylindrical metric 
$g_{Z}$ and the hyperbolic metric $g_{H^{2}}$. This structure allows us to focus on their $2$-dimensional cylindrical 
analogues. In Eddington-Finkelstein coordinates, these cylindrical versions can be brought to a canonical form 
$ds^{2}=g_{\nu\nu}d\nu^{2}+2g_{\nu r}d\nu dr$, from which the radial geodesic equations 
(Section~\ref{sec:General-Equation-forGeodesics}) can be systematically derived. Renaming the coordinates as 
$\nu=\varphi$ and $r=T$, the connection to Misner spacetime becomes manifest. 

\medskip
In each case, the base manifold \( M_Z \) of the warped product structure is conformally related to a region of 
Minkowski spacetime and admits a Killing vector field whose orbits closely resemble those of the boost-generating 
vector field \( X = x\,\partial_t + t\,\partial_x \) in two-dimensional Minkowski space. This Killing vector field generates 
a one-parameter group \( G \) of isometries, some of whose orbits are timelike. A discrete subgroup 
\( G_0 \subset G \) acts properly and freely on \( M_Z \), resulting in a Misner-like structure on the quotient 
\( M_Z / G_0 \), where the timelike orbits of \( G \) project to closed timelike curves.\footnote{As described in the 
classic book by Hawking and Ellis~\cite{Hawking+Ellis}, Misner space is constructed by taking a quotient of half of 
Minkowski spacetime under a discrete subgroup of the one-parameter isometry group generated by the vector field 
$X = x\,\partial_t + t\,\partial_x$, which acts properly and freely in that region. A similar isometry group orbit 
structure appears in the $(t, r)$-planes of both the Schwarzschild and Reissner-Nordstr\"{o}m spacetimes, allowing 
for analogous quotient constructions in these cases as well.} From a mathematical standpoint, this conformal 
structure renders the similarities in causal properties relatively unsurprising, as they directly reflect known 
features of Misner spacetime. Nevertheless, from a physics perspective, the results are noteworthy. The spacetimes 
examined differ markedly in their physical interpretation and geometric origin---ranging from a flat vacuum solution 
of the Einstein equations, to an asymptotically flat black hole solution of the vacuum Einstein equations, to a model 
that violates the weak energy condition and requires ``exotic matter'' as its source---yet they exhibit analogous causal 
behavior. Understanding these similarities across physically distinct settings provides valuable insights into how 
causality can be broken or preserved under varying geometrical and physical conditions.

\medskip
While certain models considered here require violations of classical energy conditions, their study remains valuable as a probe of the causal structure permitted by general relativity. Specifically, the observation that vacuum and non-vacuum models can exhibit identical causal behavior suggests that these pathological features are fundamentally geometric rather than artifacts of specific matter content.

\medskip
A central result of this work is the formal proof of the isocausality of the Misner-type family at the covering level. 
In Proposition~\ref{prop:IsoAllThree} we showed that the universal covers of the Misner, pseudo-Schwarzschild, 
and pseudo-Reissner-Nordstr\"{o}m spacetimes are pairwise isocausal via explicit smooth bijections 
whose differentials preserve future-directed causal cones. This isocausality also holds on suitable 
open, causally regular regions of the compactified spacetimes where the chronology-violating sets and Cauchy horizons 
are excised. Furthermore, we identified an explicit and checkable deck-equivariance criterion under which 
the isocausality descends to the full compactified spacetimes. When the equivariance degree satisfies $|k|=1$, 
global isocausality follows. If instead $|k|>1$, or if equivariance fails altogether, then 
Corollary~\ref{cor:NoEquivariance} and Remark~\ref{rem:DefaultOneWay} show that one obtains only a 
one-way causal relation between the compactified models. 
This distinction between local/covering and global/quotient behavior is essential for 
understanding the true strength of the causal equivalences in these examples.

\medskip
In the context of ongoing research on spacetimes with closed timelike curves~\cite{Awad, Emparan, Gavassino, Roy}, 
our findings offer a unified framework to view several causality-violating spacetimes through a common geometric 
and causal lens. The work reveals the intricate relationships between these examples and motivates a broader 
investigation into the interplay between geometry, causality, and spacetime topology. For instance, it would be 
compelling to investigate whether other spacetime models, such as the pseudo-Kerr spacetime~\cite{Dietz}---a 
rotating generalization of the pseudo-Schwarzschild metric---also belong to this family of causality-violating, 
Misner-type spacetimes. We view this article as a step toward a broader classification of causality-violating spacetimes. 
The warped-product ansatz and the explicit isocausality mappings presented here suggest that tools from 
conformal geometry and causal theory can provide a systematic framework for organising, comparing, and 
understanding causal structures in a wide range of gravitational models.

~
\begin{acknowledgement*}
Parts of this work were completed during my time as a member of Richard Schoen's 
research group at the University of California, Irvine, and at the Simons Center for Geometry and Physics, Stony Brook University.
The research presented here originated from a project 
conducted under the supervision of Kip S. Thorne 
during 2016--2018. I am profoundly grateful to Kip S. Thorne for 
his invaluable advice, guidance, supportive remarks, and 
encouragement throughout the course of this work. I would
also like to express my deep appreciation for his inspiring spirit
of intellectual adventure in research and the insightful discussions,
particularly on Misner space, which have been truly indispensable
to this work. I also wish to thank the referees for their careful reading and constructive comments, which have helped improve the manuscript.
\end{acknowledgement*}


\begin{thebibliography}{10}

\bibitem[1]{Awad}A. Awad and S. Eissa. Lorentzian Taub-NUT spacetimes: Misner string charges and the first law. Phys. Rev. D, 105 (2022), 124034.

\bibitem[2]{Chandrasekhar}S. Chandrasekhar. The Mathematical Theory
of Black Holes. Clarendon Press, Oxford, England and Oxford University
Press, New York (1983).

\bibitem[3]{Dietz}J. Dietz, A. Dirmeier and M. Scherfner. Geometric
analysis of particular compactly constructed time machine spacetimes.
J. Geom. Phys., 62 (2012), 1273--1283.

\bibitem[4]{Emparan}R. Emparan and M. Toma\v{s}evi\'{c}. Quantum backreaction
on chronology horizons. J. High Energ. Phys., 2022, 182 (2022).

\bibitem[5]{Frolov}V. Frolov and I. Novikov. Black Hole Physics:
Basic Concepts and New Developments. Springer Dordrecht (1998).

\bibitem[6]{Garcia-Parrado}A. Garc\'{i}a-Parrado and J. M. M. Senovilla.
Causal relationship: A new tool for the causal characterization of
Lorentzian manifolds. Class. Quantum Grav. 20 (2003), 625--664.

\bibitem[7]{Garcia-Parrado_M. Sanchez} A. Garc\'{i}a-Parrado and
M. S\'{a}nchez. Further properties of causal relationship: Causal
structure stability, new criteria for isocausality and counterexamples.
Class. Quantum Grav. 22 (2005), 4589--4619.

\bibitem[8]{Gaudin}M. Gaudin, V. Gorini, A. Kamenshchik, U. Moschella
and V. Pasquier. Gravity of a static massless scalar field and a limiting
Schwarzschild-like geometry.  Int. J. Mod. Phys. D 15 (2006).

\bibitem[9]{Gavassino}L. Gavassino. Life on a closed timelike curve. 
Class. Quantum Grav. 42 (2025), 015002.

\bibitem[10]{Goehring}R. G\"{o}hring. Kosmologie der allgemeinen
Relativit\"{a}tstheorie. Physikalischer Verein, Frankfurt am Main
(2010).

\bibitem[11]{Hawking+Ellis}S. W. Hawking and G. F. R. Ellis. The Large
Scale Structure of Space-Time. Cambridge University Press, Cambridge
(1973).

\bibitem[12]{Kim and Thorne}S. Kim and K. S. Thorne. Do vacuum fluctuations
prevent the creation of closed timelike curves? Phys. Rev. D, 43 (1991), 3929.

\bibitem[13]{Konkowski} D. A. Konkowski and L. C. Shepley. Stability
of Two-Dimensional Quasisingular Space-Times. Gen. Relativ. Gravit. 14, No. 1 (1982).

\bibitem[14]{Misner}C. W. Misner. Taub-NUT Space as a Counterexample
to Almost Anything. Relativity Theory and Astrophysics I: Relativity
and Cosmology (J. Ehlers Ed.), Lectures in Applied Mathematics 8 A.M.S.
(1967).

\bibitem[15]{Misner_Thorne}C. W. Misner, K. S. Thorne and J. A. Wheeler.
Gravitation. W. H. Freeman and Company, NY, San Francisco (1973).

\bibitem[16]{Ori}A. Ori. Formation of closed timelike curves in a
composite vacuum-dust asymptotically flat spacetime. Phys. Rev. D,
76 (2007), 044002.

\bibitem[17]{Ori Misner space}D. Levanony and A. Ori. Extended time-travelling
objects in Misner space. Phys. Rev. D, 83 (2011), 044043.

\bibitem[18]{Rieger - Topologies of maximally extended non-Hausdorff Misner Space}N.
E. Rieger. Topologies of maximally extended non-Hausdorff Misner Space.
arXiv:gr-qc/2402.09312 (2024), digital preprint of the 2016 original
version.

\bibitem[19]{Roy}D. Roy, A. Dutta and S. Chakraborty. Geodesic Motion and Particle Confinements in Cylindrical Wormhole Spacetime: Exploring Closed Timelike Curves. Int. J. Mod. Phys. A (2025).

\bibitem[20]{Sharan}P. Sharan. Spacetime, Geometry and Gravitation.
Springer, Birkh\"auser Basel (2009).

\bibitem[21]{Thorne CTC}K. S. Thorne. Closed Timelike Curves. General
Relativity and Gravitation; Proceedings of the 13th International
Conference on General Relativity and Gravitation, Institute of Physics
Publishing, Bristol, England (1993), 295--315.

\bibitem[22]{Thorne}K. S. Thorne. Misner Space as a Prototype for
Almost Any Pathology. Directions in General Relativity Eds. B. L.
Hu et al., Cambridge University Press, Cambridge (1993), 333--346.
\end{thebibliography}
\end{document}